\documentclass[10pt,journal,compsoc]{IEEEtran}
\usepackage{algorithm,amsbsy,amsmath,amssymb,epsfig,bbm,mathrsfs,multirow,amsthm,xcolor}
\usepackage{epstopdf}
\usepackage{graphicx,caption,subcaption}
\usepackage{setspace}
\usepackage{algorithmic}
\usepackage{subcaption,float}
\usepackage[hidelinks]{hyperref}

\newtheorem{definition}{Definition}
\newtheorem{lemma}{Lemma}
\newtheorem{theorem}{\textbf{\textsc{Theorem}}}

\ifCLASSOPTIONcompsoc
\usepackage[nocompress]{cite}
\else
\usepackage{cite}
\fi
\ifCLASSINFOpdf
\else
\fi

\usepackage{setspace}

\begin{document}
\title{Federated Learning Meets Contract Theory: Energy-Efficient Framework for Electric Vehicle Networks}
\author{Yuris Mulya Saputra, \emph{Student Member, IEEE},  Diep N. Nguyen, \emph{Senior Member, IEEE}, Dinh Thai Hoang, \emph{Member, IEEE}, Thang Xuan Vu, \emph{Member, IEEE}, Eryk Dutkiewicz, \emph{Senior Member, IEEE}, \\ and Symeon Chatzinotas, \emph{Senior Member, IEEE}
\IEEEcompsocitemizethanks{\IEEEcompsocthanksitem Y.~M.~Saputra is with the
	School of Electrical and Data Engineering, University of Technology Sydney, Sydney, NSW 2007, Australia (e-mail: yurismulya.saputra@student.uts.edu.au) and Department of Electrical Engineering and Informatics, Universitas Gadjah Mada, Yogyakarta 55281, Indonesia.
	D.~N.~Nguyen, D.~T.~Hoang, and E.~Dutkiewicz are with the
	School of Electrical and Data Engineering, University of Technology Sydney, Sydney, NSW 2007, Australia (e-mail: \{hoang.dinh, diep.nguyen, eryk.dutkiewicz\}@uts.edu.au).
	T.~X.~Vu and S.~Chatzinotas are with the Interdisciplinary Centre for Security, Reliability and Trust, University of Luxembourg, 1855 Luxembourg City, Luxembourg (e-mail: \{thang.vu, symeon.chatzinotas\}@uni.lu).
	Preliminary results in this paper have been presented at the IEEE GLOBECOM Conference, 2019~\cite{Saputra:2019}.
    }
    \thanks{}}


\IEEEtitleabstractindextext{%
\begin{abstract}

In this paper, we propose a novel energy-efficient framework for an electric vehicle (EV) network using a contract theoretic-based economic model to maximize the profits of charging stations (CSs) and improve the social welfare of the network. Specifically, we first introduce CS-based and CS clustering-based decentralized federated energy learning (DFEL) approaches which enable the CSs to train their own energy transactions locally to predict energy demands. In this way, each CS can exchange its learned model with other CSs to improve prediction accuracy without revealing actual datasets and reduce communication overhead among the CSs. Based on the energy demand prediction, we then design a multi-principal one-agent (MPOA) contract-based method. In particular, we formulate the CSs' utility maximization as a non-collaborative energy contract problem in which each CS maximizes its utility under common constraints from the smart grid provider (SGP) and other CSs' contracts. Then, we prove the existence of an equilibrium contract solution for all the CSs and develop an iterative algorithm at the SGP to find the equilibrium. Through simulation results using the dataset of CSs' transactions in Dundee city, the United Kingdom between 2017 and 2018, we demonstrate that our proposed method can achieve the energy demand prediction accuracy improvement up to 24.63\% and lessen communication overhead by 96.3\% compared with other machine learning algorithms. Furthermore, our proposed method can outperform non-contract-based economic models by 35\% and 36\% in terms of the CSs' utilities and social welfare of the network, respectively.
	
\end{abstract}

}

\maketitle

\IEEEdisplaynontitleabstractindextext

\IEEEpeerreviewmaketitle

\IEEEraisesectionheading{\section{Introduction}\label{sec:Int}}

\IEEEPARstart{A}ccording to the International Energy Agency's latest forecast, electric vehicles (EVs) including battery EVs and plug-in hybrid EVs will take over the existence of conventional transportation systems by 2030 with 255 millions EVs on the road~\cite{iea:2019}. Thanks to the sustainability, the EVs have attracted worldwide recognition to provide not only high energy efficiency, but also low gas emissions and low-cost oil usage~\cite{Wang:2016}. Nonetheless, the increasing adoption of EVs in the near future, which creates a massive number of energy demands, may encounter an inherent yet challenging problem for the charging station providers (CSPs) to sustain effective energy services at their charging stations (CSs). For example, in an unexpected situation when a vast number of EVs urgently need to charge the battery at once, the CSPs and their corresponding CSs may suffer from a high energy transfer cost from the smart grid provider (SGP) and heavy energy transfer congestion from the CSs to the EVs, respectively~\cite{Lopes:2011}. Moreover, due to the dynamic energy charging demands from the EVs in practice, the CSs may experience underestimation or overestimation of energy supply for the EVs. Hence, the CSPs require to make an effective economic model to maximize their profits, i.e., utilities, in the cost-effective energy transfer activity based on the dynamic energy demands.

To deal with the dynamic energy demands and optimize energy efficiency, we need to accurately predict energy demands from the EVs. Among few works that investigate energy demand prediction for EV networks, the authors in~\cite{Majidpour:2015, Chis:2017, Ma:2017, Fukushima:2018, Lopez:2019} proposed machine learning-based methods, i.e., k-nearest neighbor (kNN), online reinforcement learning, multiple-regression, shallow neural network (SNN), and deep neural network (DNN), to improve energy demand prediction accuracy at specific CSs/EVs. 
Nevertheless, these approaches may not be useful for the whole EV network since they evaluate the prediction individually at each EV/CS. Moreover, each EV/CS generally has limited number of energy request data, and thus they are not sufficient to predict the energy demands accurately. To this end, it is crucial to leverage shared information or global models to predict the energy demands for further prediction accuracy improvement of the whole network.

From the accurate prediction of energy demands, the CSs can acquire meticulous energy transfer from the SGP through economic models to satisfy the EVs' energy demands as well as optimize the CSs' profits in the EV network. For example, a collaborative scheme among CSs to distribute energy resources for EV charging was proposed in~\cite{Gusrialdi:2017, Wang:2018, Wang:2019}. However, those previous works assumed that all the CSs are owned by a single CSP. In reality, CSs may belong to different CSPs, and thus there exists non-collaborative game among the CSPs. Specifically, the CSPs compete through their CSs to attract the SGP and upcoming EVs for their own profit maximizations. 

In~\cite{Lee:2015, Yoon:2016, Yang:2016}, the authors presented non-collaborative game approaches in which CSs/EVs act as players to maximize their profits considering the influence of their neighboring CSs/EVs. 
In this case, all above approaches only work when the participating entities provide full knowledge (referred to as \emph{information symmetry}) for the energy transfer activity. 
In practice, the SGP generally keeps its energy capacity as a private information (referred to as \emph{information asymmetry})~\cite{Chen:2019}, and thus the above approaches are not applicable. In addition, the above game theoretic-based methods only consider non-negotiable mechanism. Specifically, the SGP first informs the fixed price to the CSs. Then, the CSs need to adjust their amount of requested energy based on that given information without accounting for any price negotiation between the SGP and CSs. This inflexible process may reduce the CSs' expected profits in the energy transfer process. Thus, it will be challenging to discover an appropriate energy transfer policy which maximizes the profits for both the SGP and CSs. 

Given the above, in this work we investigate a contract theoretic-based economic model, an effective incentive mechanism leveraging flexible common agreements between the participating entities under the information asymmetry~\cite{Bolton:2005}. For example, the authors in \cite{Zhang:2018, Chen2:2019, Zhang3:2018} have recently proposed the contract-based approaches considering the incomplete information between a \emph{principal} and multiple \emph{agents} in the energy transfer activity. In this case, the principal offers contracts containing energy-payment bundle while the agents are responsible to accept or reject the offered contracts. Utilizing the contract theoretic-based method, we can provide fair and efficient contract negotiation for the principal and agents ahead of the actual energy transfer activity. If the contracts are accepted, the principal can maximize its utility while satisfying the individual rationality (IR) and incentive compatibility (IC) conditions from the agents{\footnote{The IR conditions play important roles in ensuring that the agents always achieve positive utilities. Meanwhile, the IC conditions guarantee maximum utilities of the agents when the best contracts from the principal are used.}. However, since all the aforementioned works assume only one principal, e.g., one CS, in the power market, they are lack of competition to attract at least one agent. In practice, an agent, e.g., an EV or the SGP, may be tempted to receive several contracts from different principals concurrently. Hence, it is necessary to develop an effective contract-based policy for multiple principals, e.g., the competing CSs, to maximize their profits considering the agent's IR and IC constraints and enhance the social welfare of the EV network.     

This work aims to develop a novel energy-efficient framework for an EV network of CSs using a multi-principal one-agent (MPOA) contract-based economic model~\cite{Bernheim:1986}. 
To that end, we first introduce a novel CS-based decentralized federated energy learning (DFEL) method where each CS can train its transaction dataset locally and send its learned model only, i.e., local gradient, to other CSs via the internet for global model update. In this way, we can enhance the upcoming energy demand prediction accuracy and decrease the communication overhead among the CSs. To efficiently minimize the biased prediction, i.e., the problem which is influenced by the combination of imbalanced features and labels in one dataset, we incorporate the CS clustering-based DFEL solution. Particularly, this solution can categorize the CSs into several groups prior to the independent execution of learning algorithm in each group. This aims to reduce the dataset dimension based on the important feature classification~\cite{Li:2018}, and thus further improve the prediction accuracy and the learning speed of the CSs.

Based on the energy demand prediction from the CS-based DFEL method, we formulate the CSs' utility maximization as a non-collaborative energy contract optimization problem leveraging the MPOA contract policy. Specifically, each CS maximizes its utility under the SGP's IR and IC constraints as well as other CSs' contracts to improve the social welfare of the EV network. This problem is implemented at the SGP because each CS only shares energy demand information with the SGP to minimize private information disclosure with other CSs. Then, we prove the existence of equilibrium contract solution for all the CSs and develop an iterative energy contract algorithm which can obtain the equilibrium. In this case, all the CSs achieve maximum utilities when the optimal contracts from them (which meet IC and IR constraints of the SGP) are applied. This equilibrium solution can achieve the social welfare within 9\% of that obtained by the information-symmetry energy contract, i.e., when each CS completely knows the current energy demands of other CSs and the true type of the SGP. Through simulation results, we show that our proposed method can obtain 24.63\% higher than other centralized machine learning methods~\cite{Boutaba:2018} in terms of the energy demand prediction accuracy. Furthermore, the proposed method can reduce the communication overhead by 96.3\% compared with them. We also demonstrate that our proposed method can increase the utilities of the CSs up to 35\% and the social welfare by 36\% compared with non-contract-based economic models~\cite{Tang:2014} in the energy transfer activity.

The main contributions are summarized as follows:

\begin{itemize}
	
	\item We introduce the machine learning-based methods that utilize the CS-based DFEL and the CS clustering-based DFEL algorithms to improve the energy demand prediction accuracy, reduce the communication overhead for CSs, and boost the learning speed.
	
	\item We formulate the non-collaborative problem that leverages the MPOA contract-based approach to maximize the utilities of the CSs and enhance the social welfare of the EV network.
	
	\item We develop the iterative energy contract algorithm which can find the equilibrium contract solution. We show that the final solution of the algorithm can achieve less than 9\% gap from the information-symmetry energy contract solution.
	
	\item We perform comprehensive simulations to evaluate the proposed method using the actual CS dataset in Dundee city, the United Kingdom. These results provide insightful information to help the CSs in designing the efficient learning methods and obtaining the best contracts.
	
\end{itemize}
The rest of this paper is organized as follows. Section~\ref{sec:EV_scheme} describes the energy-efficient EV network framework. Section~\ref{sec:PF} discusses the problem formulation. Section~\ref{sec:EDP} introduces the DFEL algorithm along with the clustering-based method. Then, Section~\ref{sec:PFT3} and Section~\ref{sec:PS} present the problem transformation and the proposed solution of the contract-based model, respectively. The performance evaluation is given in Section~\ref{sec:PE}, and then conclusion is drawn in Section~\ref{sec:Conc}.

\section{Energy-Efficient EV Network Framework}
\label{sec:EV_scheme}

In this section, we describe the EV network system model along with the utility function definitions of the SGP and CSs.

\subsection{System Model}
\label{sec:SM}

\begin{figure}[!t]
	\centering
	\includegraphics[scale=0.28]{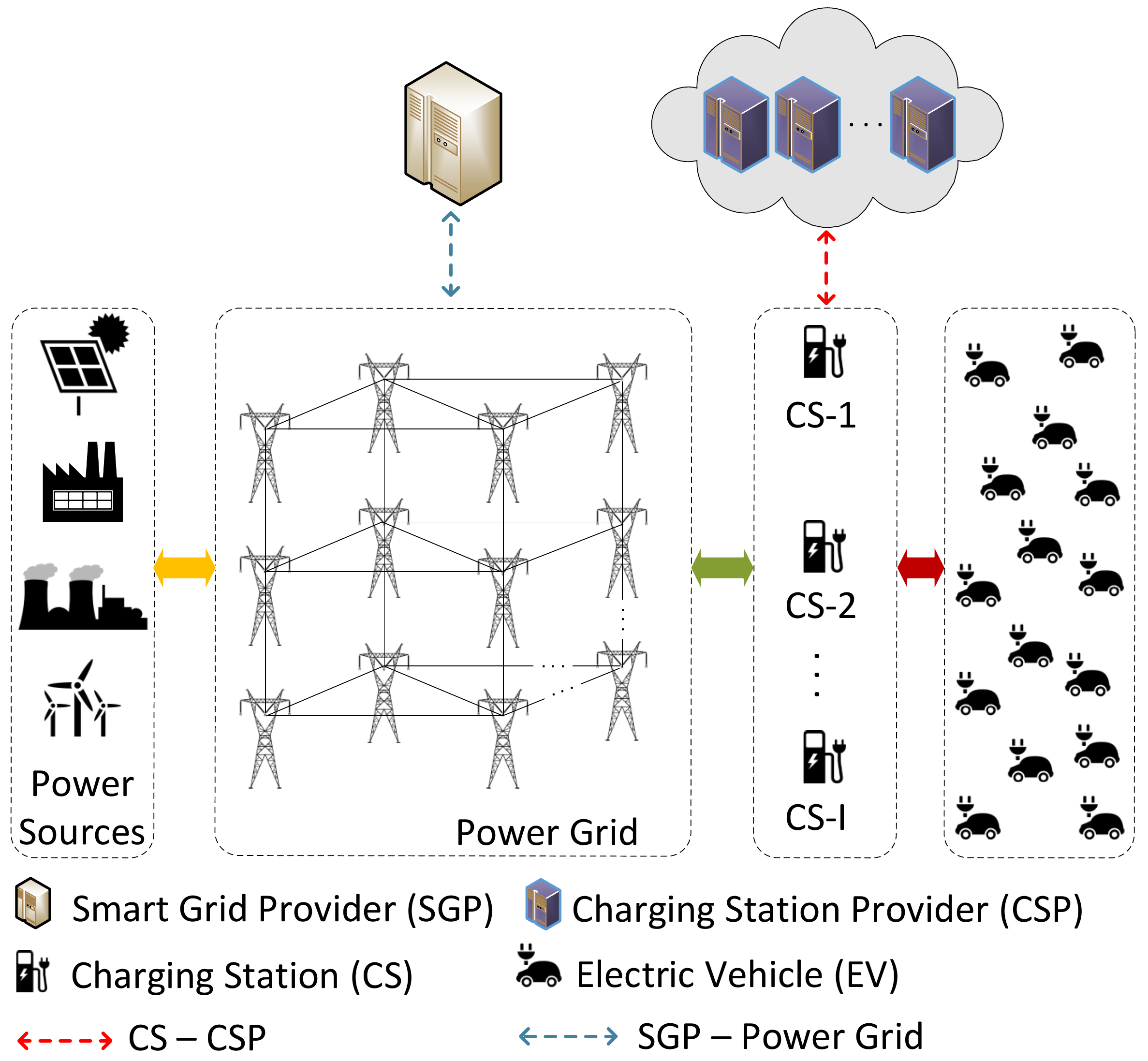}
	\caption{EV network system model.}
	\label{fig:Edge_Architecture}
\end{figure}

The EV network system model is illustrated in Fig.~\ref{fig:Edge_Architecture}. In particular, multiple CSPs orchestrate
CSs to obtain the energy from the power sources through a power grid system (which is controlled by an SGP) prior to EV charging at the CSs. For a particular period, each CS utilizes a record file which involves the EV charging transactions including CS and EV identifiers, transaction date, transaction time, and energy usage to keep track of all charging activities of the served EVs. To deal with the dynamic charging of the EVs, the CS will capture and update the record file regularly, e.g., every day~\cite{Saputra:2019}. Then, the CS can use this file to predict its energy demand for the upcoming period. Nonetheless, the number of charging transactions at each CS may not be sufficient to accurately predict the energy demand. Consequently, each CS requires to exchange its learned model, i.e., local gradient, to other CSs right after completing the energy demand learning over its local charging transaction dataset.
Upon predicting the energy demand, each CS can offer a contract containing predicted energy demand and offered payment to the SGP at a pre-defined time interval, i.e., every day or every week, to maximize its utility for energy transfer activity. 

We denote $\mathcal{I} = \{1,\ldots,i,\ldots,I\}$ and $\mathcal{N}_i = \{1,\ldots,n,\ldots,N_i\}$ to be the set of CSs in the EV network and the set of EV charging transactions at CS-$i$, respectively. Then, we specify the total energy demand and the energy demand of transaction $n$ at CS-$i$ as ${\xi}_{i}$ and $\xi_{i}^{n}$, respectively. To perform energy transfer activity for the CSs, the SGP requires to store an amount of energy at its storage with the energy capacity of $S$. We also denote $\rho_i$ and $\varrho_i$ to be the energy transfer payment for the SGP at CS-$i$ and the energy charging price per MWh for EVs at CS-$i$, respectively. To represent the SGP's private information, i.e., energy capacity~\cite{Chen:2019} and the willingness to transfer energy~\cite{Zhang:2015}, we define $\Phi = \{\phi_{min},\ldots,\phi,\ldots,\phi_{max}\}$ to be the set of the SGP's types, where $\phi_{min}$ and $\phi_{max}$ refer to the lowest and highest possible types of the SGP, respectively. In practice, the SGP may have finite number of possible types~\cite{Hu:2019}. As such, the SGP with a higher type implies a larger energy capacity and a higher willingness to transfer energy for the CSs. For example, the SGP with type 1, type 2, and type 3 represent small, intermediate, and large energy capacity, respectively. Despite that the CSs do not have any knowledge of the exact type of the SGP, they can still observe the distributions of all possible types~\cite{Zhang:2015, Hu:2019}, which is defined as $p(\phi), \forall \phi \in \Phi$. From the $p(\phi)$, each CS can offer an energy contract to the SGP according to its own energy demand and financial plan. 

To participate in the energy transfer, the SGP will generate a proportion of each CS's energy demand due to the limited energy at the SGP's storage. Particularly, we specify a vector of proportions that the SGP with type $\phi$ will transfer the requested energy to all the CSs by $\boldsymbol{\pi} = [\pi_1,\ldots,\pi_i,\ldots,\pi_I]$, where $0 \leq \pi_i \leq 1$ (as the SGP knows its own type, we do not need to consider the proportion as a function of type $\phi$). Furthermore, we define a vector of all the CSs' requested energy by $\boldsymbol{\xi}(\phi) = [\xi_1(\phi),\ldots,\xi_i(\phi),\ldots,\xi_I(\phi)], \forall \phi \in \Phi$, where $\xi_i(\phi) \geq 0, \forall i \in \mathcal{I}, \forall \phi \in \Phi$. Correspondingly, for all energy demands, all the CSs will offer a vector of payments to the SGP which is denoted by $\boldsymbol{\rho}(\phi) = [\rho_1(\phi),\ldots,\rho_i(\phi),\ldots,\rho_I(\phi)], \forall \phi \in \Phi$, where $\rho_i(\phi) \geq 0, \forall i \in \mathcal{I}, \forall \phi \in \Phi$. Intuitively, the payment increases when the requested energy increases and vice versa~\cite{Chen2:2019,Su:2019}, i.e., $\frac{d\rho_i(\phi)}{d\xi_i(\phi)} > 0, \frac{d\xi_i(\phi)}{d\rho_i(\phi)} > 0,  \forall i \in \mathcal{I}, \forall \phi \in \Phi$.

\subsection{Utility Functions of the SGP and CSs}

Given $\boldsymbol{\xi}(\phi)$ and $\boldsymbol{\rho}(\phi)$, the utility function of the SGP with type $\phi$ is expressed by
\begin{equation}
\label{eqn:for3}
\begin{aligned}
\phi G(\boldsymbol{\pi},\boldsymbol{\rho}(\phi)) - C(\boldsymbol{\pi},\boldsymbol{\xi}(\phi)),
\end{aligned}
\end{equation}
where $\phi$ in the first component is to show the true type of the SGP when receiving certain contracts from CSs. Moreover, $G(\boldsymbol{\pi},\boldsymbol{\rho}(\phi))$ and $C(\boldsymbol{\pi},\boldsymbol{\xi}(\phi))$ represent the gain and cost functions, respectively. Specifically, we apply a natural logarithm function which is widely used to quantify the gain of energy providers~\cite{Chen2:2019,Lee2:2015}. Thus, the gain function can be written by
\begin{equation}
\label{eqn:for4}
\begin{aligned}
G(\boldsymbol{\pi},\boldsymbol{\rho}(\phi)) = \text{ ln}\Bigg(1 + \sum_{i=1}^I \pi_i\rho_i(\phi)\Bigg).
\end{aligned}
\end{equation}
From Eq.~(\ref{eqn:for4}), the gain function follows the law of diminishing returns, i.e., the function increases as the received payments from the CSs increase and the marginal utility keeps decreasing until the function achieves the saturation point~\cite{Samuelson:2005}. 
Meanwhile, the cost function of the SGP to transfer energy for all CSs can be formulated by
\begin{equation}
\label{eqn:for5}
\begin{aligned}
C(\boldsymbol{\pi},\boldsymbol{\xi}(\phi)) &=  \zeta\sum_{i=1}^{I}\pi_i\xi_i(\phi),
\end{aligned}
\end{equation}
where $\zeta > 0$ is the energy transfer cost per unit. 

According to Eq.~(\ref{eqn:for3}), we need to maximize the utility of the SGP with type $\phi$ as the following optimization problem:
\begin{equation}
\label{eqn:for3b}
\begin{aligned}
(\mathbf{P}_1) \phantom{10} & \max_{\boldsymbol{\pi}} \phantom{5} \phi G(\boldsymbol{\pi},\boldsymbol{\rho}(\phi)) -  C(\boldsymbol{\pi},\boldsymbol{\xi}(\phi)).
\end{aligned}
\end{equation}
\begin{eqnarray}
\text{ s.t. } \quad &\overset{I}{\underset{i=1}{\sum}} \pi_i \xi_i(\phi) \leq S(\phi), \label{eqn:for3c}  \\
&0 \leq \pi_i \leq 1, \forall i \in \mathcal{I}, \label{eqn:for3c2} 
\end{eqnarray}
where $S(\phi)$ represents the energy capacity for the SGP with type $\phi$. The constraint (\ref{eqn:for3c}) specifies that the total amount of actual energy demands from CSs must not exceed the energy capacity of the SGP. Based on $(\mathbf{P}_1)$, we can find the optimal $\boldsymbol{\hat \pi}$, where $\boldsymbol{\hat \pi} = [{\hat \pi}_1,\ldots,{\hat \pi}_i,\ldots,{\hat \pi}_I]$. This $\boldsymbol{\hat \pi}$ can be obtained straightforwardly due to the convexity of the SGP's utility function. Specifically, the gain function follows the law of diminishing returns with a positive, strictly concave, continuously differentiable, and strictly increasing utility function~\cite{Bernheim:1986,Lee2:2015}, i.e., $G(\boldsymbol{\pi}, 0) = 0$, $G'(\boldsymbol{\pi},\boldsymbol{\rho}(\phi)) > 0$, and $G''(\boldsymbol{\pi},\boldsymbol{\rho}(\phi)) < 0$ for all $\boldsymbol{\rho}(\phi)$. In addition, the cost function holds a linear function which is a positive and strictly increasing function~\cite{Bolton:2005}. Considering $\boldsymbol{\hat \pi}$, the CSs compensate energy transfer payments to obtain certain amount of energy from the SGP. As a result, the expected utility function of a CS-$i$, which specifies the expected profit of the CS-$i$ obtained from the energy charging activity to EVs, can be expressed as
\begin{equation}
\label{eqn:for6}
\begin{aligned}
U_i(\boldsymbol{\rho}(\phi), \boldsymbol{\xi}(\phi)) = {\hat \pi}_i\sum_{\phi=\phi_{min}}^{\phi_{max}}\Big(\varrho_i \xi_i(\phi) - \rho_i(\phi)\Big) p(\phi),
\end{aligned}
\end{equation}
where $\overset{\phi_{max}}{\underset{\phi=\phi_{min}}{\sum}} p(\phi)=1$, ${\hat \pi}_i\Big( \varrho_i \xi_i(\phi) - \rho_i(\phi)\Big)$ is the actual utility function of CS-$i$ which implements the possible type $\phi$, and $\varrho_i > 0$ is the gain parameter representing the energy price unit per MWh for served EVs at CS-$i$. Furthermore, the use of $\sum(.)$ indicates that the expected utility of each CS relies on the SGP's possible type distribution. 

From the optimal proportions $\boldsymbol{\hat \pi}$, and contracts $\Big(\boldsymbol{\rho}(\phi), \boldsymbol{\xi}(\phi)\Big), \forall \phi \in \Phi$, we can derive the social welfare of the network as the total ulitilies of the SGP with type $\phi$ and all the participating CSs, which is
\begin{equation}
\label{eqn:for9rev3}
\begin{aligned}
U_{SW}(\phi) = \phi\text{ ln}\Bigg(1 + \sum_{i=1}^I {\hat \pi_i}{ \rho}_i(\phi)\Bigg) - \zeta\sum_{i=1}^{I}{\hat \pi}_i{\xi}_i(\phi) \\ + \sum_{i=1}^I{\hat \pi}_i\Big( \varrho_i {\xi}_i(\phi) - {\rho}_i(\phi)\Big), \forall \phi \in \Phi.
\end{aligned}
\end{equation}

\section{Problem Formulation}\label{sec:PF}

Based on the utility functions in Eq.~(\ref{eqn:for3})-(\ref{eqn:for6}), we can formulate the energy contract optimization problem utilizing $I$ CSs as the principals and one SGP as the agent in the EV network. The objective is to maximize the expected profits, i.e., utilities, of all CSs independently while satisfying the utility requirements of the SGP. First, as a benchmark case, consider that each CS has current energy demand information from other CSs and the true type, i.e., energy capacity, of the SGP. In this case, the offered contracts $\Big(\boldsymbol{\rho}(\phi),\boldsymbol{\xi}(\phi)\Big)$ for type $\phi$ must satisfy individual rationality (IR) requirements as defined in Definition~\ref{DefV1}. 
\begin{definition}{}\label{DefV1}
	Individual rationality constraint: The SGP must obtain a non-negative utility, i.e.,
	\begin{equation}
	\label{eqn:for7a}
	\begin{aligned}
	\phi G(\boldsymbol{\hat \pi},\boldsymbol{\rho}(\phi)) - C(\boldsymbol{\hat \pi},\boldsymbol{\xi}(\phi)) \geq 0, \forall \phi \in \Phi.
	\end{aligned}
	\end{equation}
\end{definition}
Given the SGP's type $\phi$, we can formulate the information-symmetry energy contract optimization problem for each CS-$i$ to maximize its utility as follows:
\begin{equation}
\label{eqn:for90}
\begin{aligned}
(\mathbf{P}_2) \phantom{10} & \underset{\boldsymbol{\rho}(\phi),\boldsymbol{\xi}(\phi)}{\text{max}} \phantom{5} {\hat \pi}_i\Big( \varrho_i \xi_i(\phi) - \rho_i(\phi)\Big), \forall i \in \mathcal{I},
\end{aligned}
\end{equation}
\begin{eqnarray}
\text{ s.t. } \quad &\overset{I}{\underset{i=1}{\sum}} {\hat \pi}_i \xi_i(\phi) \leq S(\phi), \label{eqn:for90a} \\
& \phi G(\boldsymbol{\hat \pi},\boldsymbol{\rho}(\phi)) - C(\boldsymbol{\hat \pi},\boldsymbol{\xi}(\phi)) \geq 0 \label{eqn:for90b}.
\end{eqnarray}

In practice, the SGP keeps its type, i.e., the energy capacity and the willingness to transfer energy, private and unobservable to the CSs. As a result, it leads to the information asymmetry between the SGP and the CSs. Due to this information asymmetry, in addition to the IR constraints, the offered contracts must also satisfy the incentive compatibility (IC) constraints (as defined in the following Definition~\ref{DefV2}) to ensure the feasibility of the contracts.

\begin{definition}{}\label{DefV2}
	Incentive compatibility constraint: The SGP with type $\phi$ obtains maximum utility when it receives contracts designed for its true type $\phi$, i.e.,
	\begin{equation}
	\label{eqn:for7b}
	\begin{aligned}
\phi G(\boldsymbol{\hat \pi},\boldsymbol{\rho}(\phi)) - C(\boldsymbol{\hat \pi},\boldsymbol{\xi}(\phi)) \geq \phi G(\boldsymbol{\hat \pi},\boldsymbol{\rho}(\hat \phi)) - C(\boldsymbol{\hat \pi},\boldsymbol{\xi}(\hat \phi)), \\ \phi \neq {\hat \phi}, \forall \phi, {\hat \phi} \in \Phi,
	\end{aligned}
	\end{equation}
	where ${\hat \phi}, \forall {\hat \phi} \in \Phi$ are possible types of the SGP with ${\hat \phi} \neq \phi$.
\end{definition}
From the aforementioned IR and IC constraints, we can formulate a non-collaborative energy contract optimization problem where multiple CSPs through their CSs do not exchange the energy demand information with each other due to their selfishness and privacy concerns. Specifically, each CSP controls different CSs and offers a contract to the SGP separately. In other words, each CS-$i$ can communicate with the SGP independently using its individual contract $\Big(\rho_i(\phi),\xi_i(\phi)\Big)$. Then, the energy contract optimization problem $(\mathbf{P}_3)$ to maximize the expected utility for CS-$i$ separately at the SGP can be written as
\begin{equation}
\label{eqn:for9}
\begin{aligned}
(\mathbf{P}_3) \phantom{10} & \underset{\boldsymbol{\rho}(\phi),\boldsymbol{\xi}(\phi)}{\text{max}} \phantom{5} U_{i}(\boldsymbol{\rho}(\phi), \boldsymbol{\xi}(\phi)), \forall i \in \mathcal{I},
\end{aligned}
\end{equation}
\begin{eqnarray}
&\text{ s.t. } \quad \overset{I}{\underset{i=1}{\sum}} {\hat \pi}_i \xi_i(\phi) \leq S(\phi), \forall \phi \in \Phi, \label{eqn:for9a} \\
&\phi G(\boldsymbol{\hat \pi},\boldsymbol{\rho}(\phi)) - C(\boldsymbol{\hat \pi},\boldsymbol{\xi}(\phi)) \geq 0, \forall \phi \in \Phi, \label{eqn:for9b} \\
&\phi G(\boldsymbol{\hat \pi},\boldsymbol{\rho}(\phi)) - C(\boldsymbol{\hat \pi},\boldsymbol{\xi}(\phi)) \geq \nonumber \\
&\phi G(\boldsymbol{\hat \pi},\boldsymbol{\rho}(\hat \phi)) - C(\boldsymbol{\hat \pi},\boldsymbol{\xi}(\hat \phi)),  
\phi \neq {\hat \phi}, \forall \phi, {\hat \phi} \in \Phi \label{eqn:for9c}.
\end{eqnarray}
From the $(\mathbf{P}_3)$, the SGP's choice of ${\hat \pi}_i$ will impact the optimal contract of CS-$i$. In particular, the SGP is expected to give more proportions to a CS which requests lower energy transfer because this CS can help to increase the SGP's utility. Meanwhile, the CS can obtain higher utility when it requests a larger amount of energy transfer to the SGP. These conditions trigger a utility maximization trade-off between the SGP and CSs. In this way, each CS will compete with other CSs in a competition strategy to attract the SGP and find an equilibrium solution. Given the equilibrium contract choices of other CSs $\Big(\boldsymbol{\hat \rho}_{-i}(\phi), \boldsymbol{\hat \xi}_{-i}(\phi)\Big), \forall \phi \in \Phi$, the expected utility of CS-$i$ with equilibrium contract selection $\Big({\hat \rho}_i(\phi),{\hat \xi}_i(\phi)\Big), \forall \phi \in \Phi$ must be the maximum one among the expected utilities of the CS-$i$ as defined in Definition~\ref{DefV3}. In other words, there is no CS-$i$ which has an additional expected utility to deviate from its $\Big({\hat \rho}_i(\phi), {\hat \xi}_i(\phi)\Big), \forall \phi \in \Phi$, unilaterally.
\begin{definition}{}\label{DefV3}
	Equilibrium contract solution for non-collaborative energy contract optimization problem: The optimal contracts $\Big(\boldsymbol{\hat \rho}(\phi), \boldsymbol{\hat \xi}(\phi)\Big), \forall \phi \in \Phi$, are the equilibrium solution of the $(\mathbf{P}_3)$ if and only if the following conditions hold
	\begin{equation}
	\label{eqn:for10}
	\begin{aligned}
	U_i({\hat \rho}_i(\phi), {\hat \xi}_i(\phi), \boldsymbol{\hat \rho}_{-i}(\phi), \boldsymbol{\hat \xi}_{-i}(\phi)) \geq \\ U_i(\rho_i(\phi),\xi_i(\phi),\boldsymbol{\hat \rho}_{-i}(\phi), \boldsymbol{\hat \xi}_{-i}(\phi)), \forall \phi \in \Phi, \forall i \in \mathcal{I},
	\end{aligned}
	\end{equation}
	that satisfy the constraints (\ref{eqn:for9b}) and (\ref{eqn:for9c}).
\end{definition}

To find the optimal contracts $\Big(\boldsymbol{\hat \rho}(\phi), \boldsymbol{\hat \xi}(\phi)\Big)$ and then achieve the maximum profits of CSs while improving the social welfare of the EV network, we propose the following procedures as illustrated in Fig.~\ref{fig:contract_generation}. First, the CSs require to predict energy demands for all the CSs utilizing federated learning to enhance the prediction accuracy. Based on this accurate prediction, the CSs send initial contracts containing the predicted energy demand and offered payment information to the SGP. Then, we develop the feasible MPOA contract-based problem through transforming the $(\mathbf{P}_3)$ into simplified problem at the SGP. Based on the problem transformation, we can design the iterative energy contract algorithm as the proposed solution to find the equilibrium contract solution, i.e., optimal contracts, before sending back to the CSs.

\section{Federated Energy Learning} \label{sec:EDP}
In this section, we aim to predict energy demands of CSs locally using a CS-based DFEL method. We then extend the framework by applying a CS clustering-based method to further improve the prediction accuracy and the learning speed. All methods are applicable in the following scenarios.

\begin{figure}[t]
	\centering
	\includegraphics[scale=0.4]{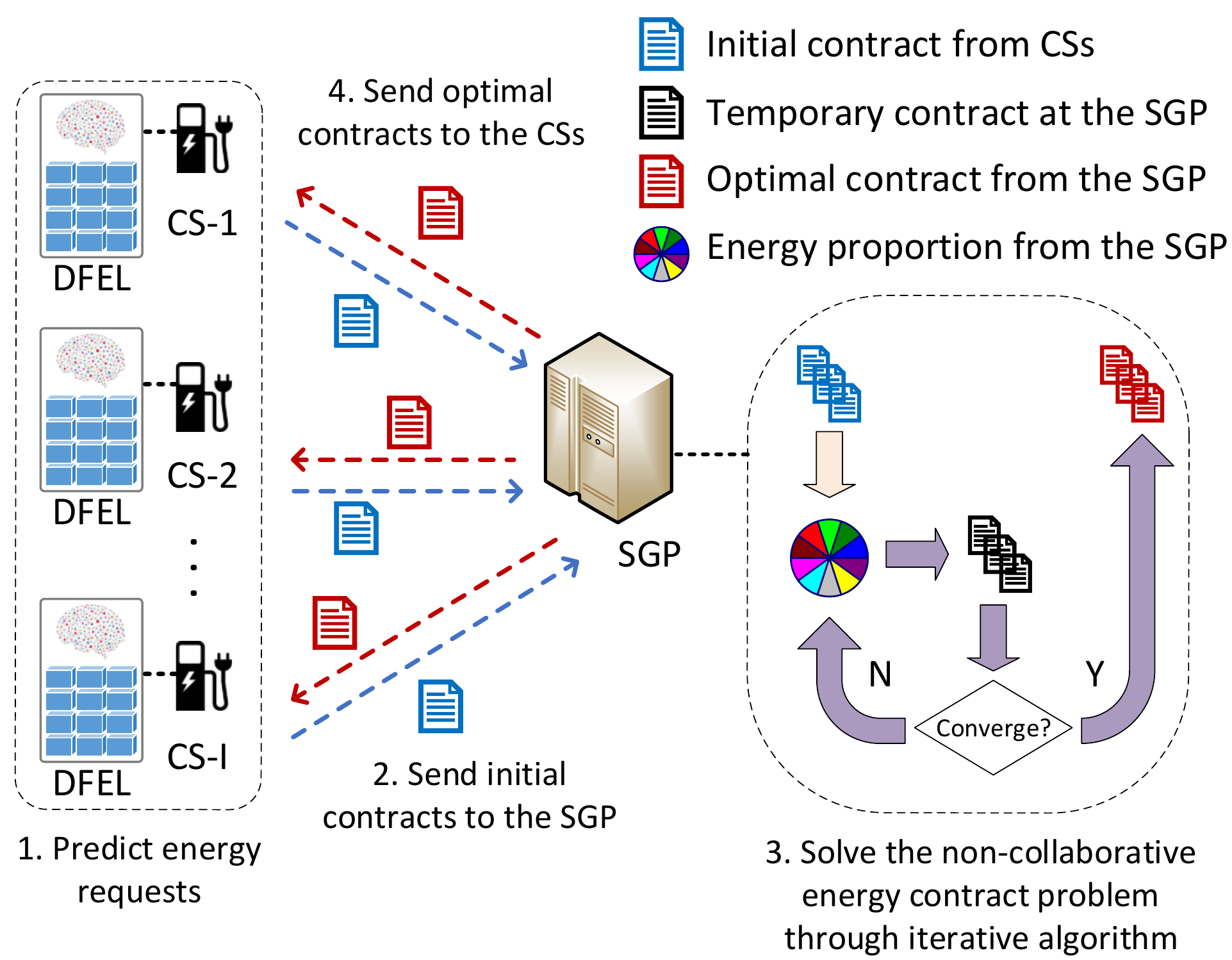}
	\caption{Illustration of the optimal contract generation process.}
	\label{fig:contract_generation}
\end{figure}

\begin{figure*}[!]
	\begin{center}
		$\begin{array}{cc} 
		\epsfxsize=2.7 in \epsffile{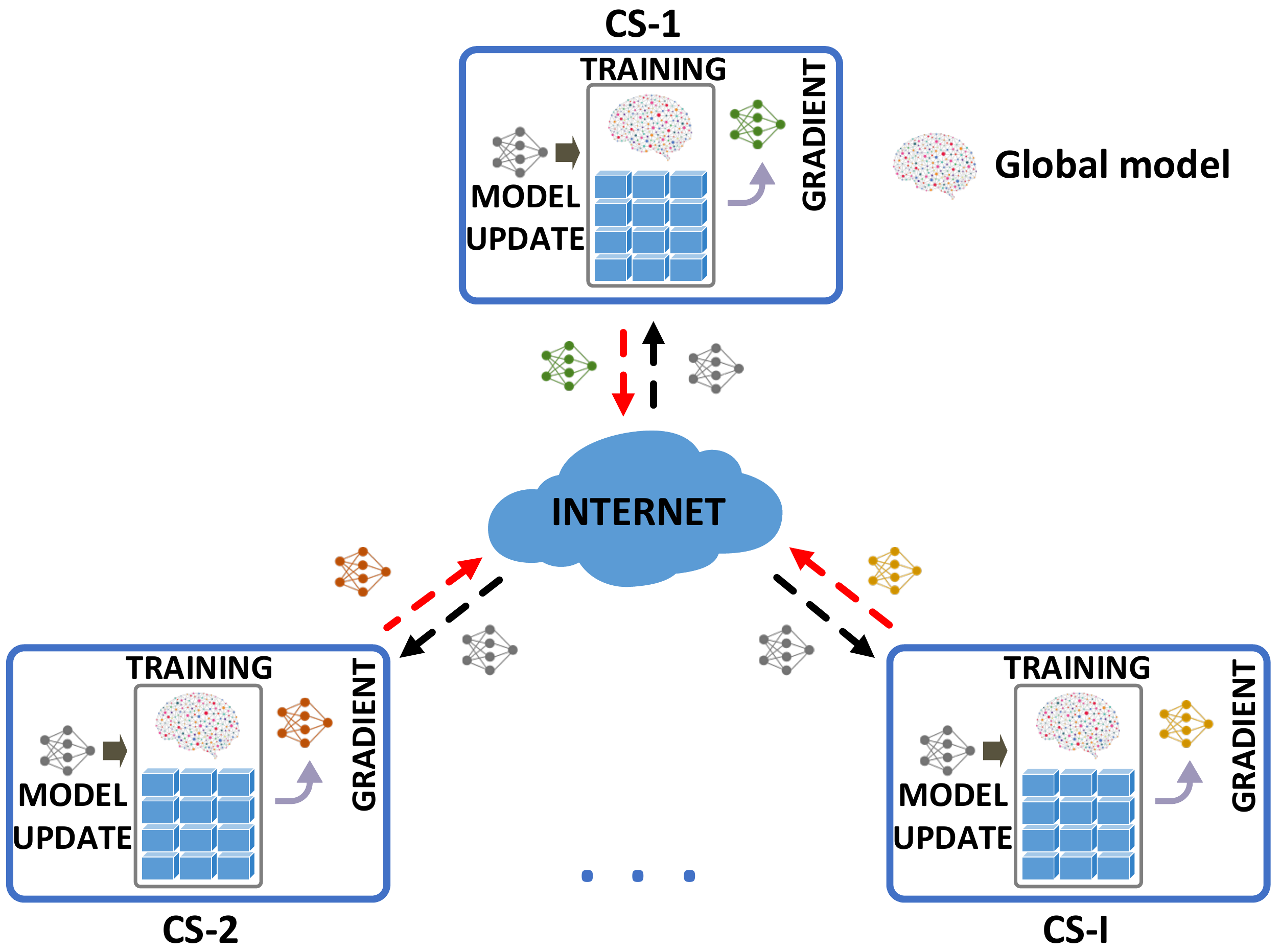} &
		\hspace*{0.4cm}
		\epsfxsize=4 in \epsffile{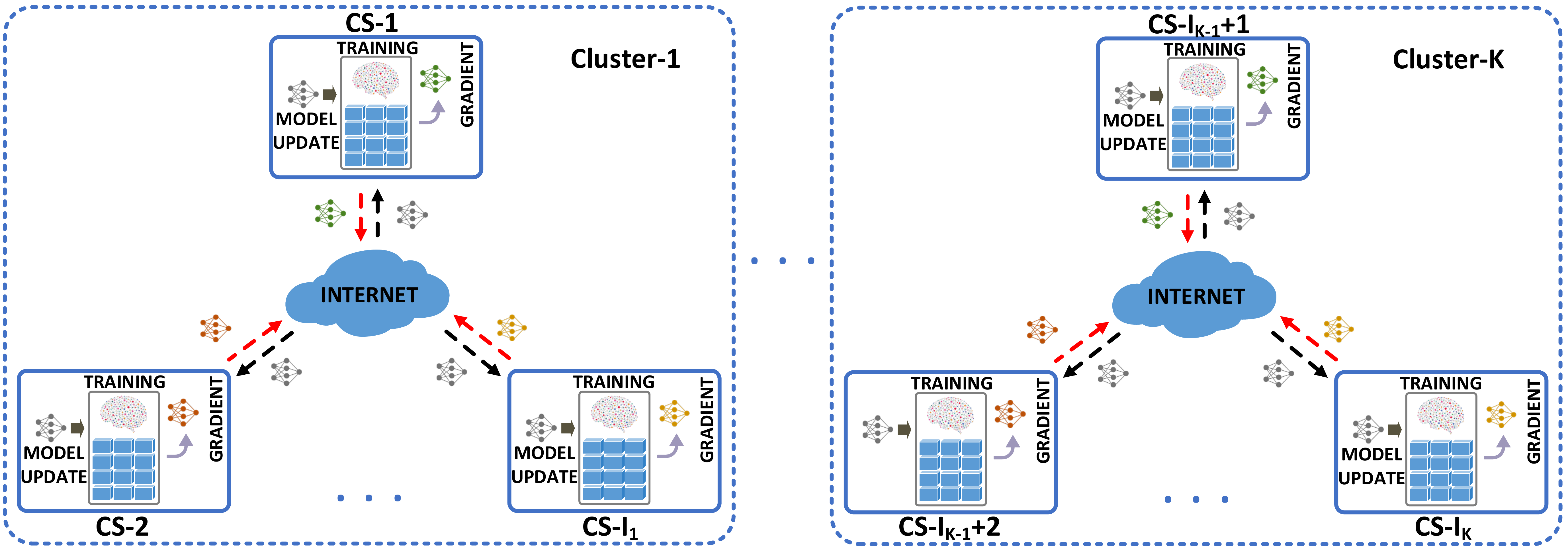} \\ [0.1cm]
		\text{\footnotesize (a) CS-based DFEL} & \text{\footnotesize (b) CS-Based DFEL with clustering}  \\ [0.2cm]
		\end{array}$
		\caption{The proposed federated energy learning.}
		\label{fig:learning_method}
	\end{center}
\end{figure*}


\subsection{CS-Based Decentralized Federated Energy Learning}
\label{subsec:EDP1}

In this method, each CS can train its local energy charging transaction dataset independently as shown in Fig.~\ref{fig:learning_method}(a). The CS can then send and receive learned models, i.e., gradient information of the DNN, only to and from other CSs, respectively. 
This is different from the conventional or centralized learning method where all the CSs send their local actual datasets to the cloud server for the learning process. Leveraging the proposed learning method, we can predict the energy demand of the CSs accurately and reduce the communication overhead as well as private information disclosure of the CSs significantly. Naturally, the CSs compete to attract incoming EVs, and thus it is beneficial to keep their economic model private from each other. Using its own and other CSs' learned models, the CS can update the global model locally~\cite{Li:2019}. Then, the CSs can use the current global model to train its own energy charging transaction dataset using the learning algorithm for the next interval. As such, the CSs operate as workers to train their EV charging transactions, i.e., $\mathbf{W}_i$,  including CS and EV identifiers, transaction date, and transaction time as the features as well as $\xi_i^n, \forall n \in \mathcal{N}_i$, $\forall i \in \mathcal{I}$ as the label, locally. 

To train the dataset at CSs, we utilize a deep learning algorithm through the DNN. In particular, consider $\mathcal{L} = \{1,\ldots,l,\ldots,L\}$ as the set of learning layers. In each layer-$l$, we have $\mathbf{A}_l$ as the global weight matrix and $\mathbf{b}_l$ as the global bias vector. Then, each CS can generate the learning input matrix based on the $\mathbf{W}_i$ of layer $l$, i.e., $\mathbf{W}_i^l$, to obtain the learning output matrix of layer-$l$ by
\begin{equation}
\label{eqn3a}
\begin{aligned}
\mathbf{\hat W}_i^{l} = \alpha_i \big(\mathbf{A}_l\mathbf{W}_i^{l} + \mathbf{b}_l\big),
\end{aligned}
\end{equation} 
where $\alpha_i$ represents a \emph{tanh} activation function~\cite{Zhang2:2018} to determine the output of DNN using the hyperbolic tangent of $\mathbf{W}_i^l$ at CS-$i$ as follows:
\begin{equation}
\label{eqn3a2}
\begin{aligned}
\alpha_i = \frac{e^{\mathbf{W}_i^l}-e^{-\mathbf{W}_i^l}}{e^{\mathbf{W}_i^l}+e^{-\mathbf{W}_i^l}}.
\end{aligned}
\end{equation}
The use of this function is to avoid zero-value gradient when the weight matrix is initialized. We also apply several hidden layers $l$, where $1 < l < L$ and $\mathbf{W}_i^{l+1} = \mathbf{\hat W}_i^l$. To deal with the overfitting and generalization error, we add a dropout layer $l_{\emph{\mbox{drop}}}$, where $l_{\emph{\mbox{drop}}} < L$, after the final hidden layer. In this case, the $l_{\emph{\mbox{drop}}}$ will randomly drop the $\mathbf{W}_i^{l_{\emph{\mbox{drop}}}}$ with a certain fraction rate. 

Considering $\mathbf{A} = [\mathbf{A}_1,\ldots,\mathbf{A}_l,\ldots,\mathbf{A}_L]$ and $\mathbf{b} = [\mathbf{b}_1,\ldots,\mathbf{b}_l,\ldots,\mathbf{b}_L]$ as the global weight and bias vectors of all layers, we can define $\boldsymbol{\psi} = (\mathbf{A}, \mathbf{b})$ as the global model for all layers. Then, we can derive the prediction error $\omega_i(\boldsymbol{\psi}^{(\tau)})$ for the time when we have checked all charging transactions of $\mathbf{W}_i$, i.e., epoch time $\tau$, at CS-$i$ by
\begin{equation}
\label{eqn3b}
\begin{aligned}
\omega_i(\boldsymbol{\psi}^{(\tau)}) = \sum_{n=1}^{N_i}\omega^{n}_i(\boldsymbol{\psi}^{(\tau)}),
\end{aligned}
\end{equation}
where $\omega_i^{n}(\boldsymbol{\psi}^{(\tau)}) = ({\hat w}_i^{n} - w_i^{n})^2$, with $w_i^{n}$ and ${\hat w}_i^{n}$ are the components of learning input matrix $\mathbf{W}_i^1$ and learning output matrix $\mathbf{\hat W}_i^L$, respectively. Based on $\omega_i(\boldsymbol{\psi}^{(\tau)})$, the local gradient of CS-$i$ at $\tau$ can be calculated by 
\begin{equation}
\label{eqn3c}
\begin{aligned}
\nabla \boldsymbol{\psi}_i^{(\tau)} = \frac{\partial \omega_i(\boldsymbol{\psi}^{(\tau)})}{\partial \boldsymbol{\psi}^{(\tau)}}.
\end{aligned}
\end{equation}
Upon obtaining $\nabla \boldsymbol{\psi}_i^{(\tau)}, \forall i \in \mathcal{I}$, the CSs exchange them for global gradient accumulation through the internet which is described by
\begin{equation}
\label{eqn3i}
\begin{aligned}
\nabla \boldsymbol{\psi}^{(\tau)} = \frac{1}{I}\sum_{i=1}^{I} \nabla \boldsymbol{\psi}_i^{(\tau)}.
\end{aligned}
\end{equation}
To guarantee that there is no outdated global model to calculate the local gradients, we apply the gradient accumulation after successfully collecting $I$ local gradients from all the CSs on the internet. Then, this accumulated global gradient is sent back to the CSs for the global model update. In this way, all the CSs can share the local gradients and update the global model $\boldsymbol{\psi}^{(\tau)}$ collaboratively to leverage the prediction accuracy. 

To obtain minimum prediction error, i.e., $\underset{\boldsymbol{\psi}}{\text{\bf min }}\omega_i(\boldsymbol{\psi})$, we use \emph{Adam} optimizer as the adaptive learning rate~\cite{Kingma:2015}, aiming at producing high robustness and achieving fast convergence to the global model. Specifically, we denote $\eta_\tau$ and $\delta_\tau$ to be the exponential moving average of the $\nabla \boldsymbol{\psi}^{(\tau)}$ and the squared $\nabla \boldsymbol{\psi}^{(\tau)}$ to obtain the variance at $\tau$, respectively. Then, we can express the update rules of $\eta_{\tau+1}$ and $\delta_{\tau+1}$ in the following equation:
\begin{equation}
\begin{aligned}
\label{eqn3d1}
\eta_{\tau+1} &= \gamma_\eta^\tau \eta_{\tau} + (1 - \gamma_\eta^\tau)\nabla \boldsymbol{\psi}^{(\tau)},\\
\delta_{\tau+1} &= \gamma_\delta^\tau \delta_{\tau} + (1 - \gamma_\delta^\tau)(\nabla \boldsymbol{\psi}^{(\tau)})^2,
\end{aligned}
\end{equation}
where $\gamma_\eta^\tau$ and $\gamma_\delta^\tau \in [0,1)$ specify $\eta_\tau$'s and  $\delta_\tau$'s steps of the exponential decays at $\tau$, respectively. To further control how often the global model is updated, we also account for the learning step $\lambda$. As such, we can update the $\lambda$ as follows:
\begin{equation}
\label{eqn3d}
\begin{aligned}
\lambda_{\tau+1} = \lambda\frac{\sqrt{1 - \gamma_\delta^{\tau+1}}}{1 - \gamma_\eta^{\tau+1}}.
\end{aligned}
\end{equation}
Finally, we can update the global model $\boldsymbol{\psi}^{(\tau+1)}$ to train $\mathbf{W}_i, \forall i \in \mathcal{I}$, for the next $\tau+1$ by
\begin{equation}
\label{eqn3e}
\begin{aligned}
\boldsymbol{\psi}^{(\tau+1)} = \boldsymbol{\psi}^{(\tau)} - \lambda_{\tau+1}\frac{\eta_{\tau+1}}{\sqrt{\delta_{\tau+1}} + \varepsilon},
\end{aligned}
\end{equation}
where $\varepsilon$ indicates a fixed value to avoid zero division when the $\sqrt{\delta_{\tau+1}}$ closes to zero. The learning process are repeated until the prediction error reaches a minimum convergence, or a pre-defined epoch time threshold $\tau_{\emph{\mbox{epoch}}}$ is obtained. To predict $\mathbf{\hat W}^{*}_i, \forall i \in \mathcal{I}$ of training dataset $\mathbf{W}_i, \forall i \in \mathcal{I}$ and new dataset at CS-$i$, $\forall i \in \mathcal{I}$, we can first generate the final global model $\boldsymbol{\psi}^*$ and then use Eq.~(\ref{eqn3a}) at all the CSs. The summary of the CS-based DFEL algorithm is presented in Algorithm~\ref{DDL-PC}.

\begin{algorithm}[]
	\caption{CS-Based DFEL Algorithm} \label{DDL-PC}
	
	\begin{algorithmic}[1] 
		
		\STATE Set $\alpha_i, \forall i \in \mathcal{I}$ and initial $\boldsymbol{\psi}^{(\tau)}$ when $\tau=0$ for all CSs
		
		\STATE Generate $\mathbf{W}_i$ containing $\xi_i^n, \forall n \in \mathcal{N}_i$, $\forall i \in \mathcal{I}$
		
		\WHILE{$\tau \leq \tau_{\emph{\mbox{epoch}}} \text{ {\bf and} } \omega_i(\boldsymbol{\psi}^{(\tau)}), \forall i \in \mathcal{I} \text{ do not converge}$}
		
		\FOR{$\forall i \in \mathcal{I}$}
		
		\STATE  $\mathbf{W}_i^{1}$ to produce $\mathbf{\hat W}_i^{L}$ at layer-$L$ using $\boldsymbol{\psi}^{(\tau)}$
		
		\STATE Compute $\omega_i(\boldsymbol{\psi}^{(\tau)})$ and $\nabla \boldsymbol{\psi}_i^{(\tau)}$
		
		\STATE Send $\nabla \boldsymbol{\psi}_i^{(\tau)}$ to the internet for global gradient accumulation
		
		\STATE Receive $\nabla \boldsymbol{\psi}^{(\tau)}$ from the internet
		
		\STATE Calculate $\boldsymbol{\psi}^{(\tau)}$
		
		\STATE Update global model $\boldsymbol{\psi}^{(\tau+1)}$
		
		\ENDFOR
		
		\STATE $\tau \leftarrow \tau+1$
		
		\ENDWHILE
		
		\FOR{$\forall i \in \mathcal{I}$}
		
		\STATE Predict $\mathbf{\hat W}^{*}_i$ for the next-interval energy demands using $\mathbf{W}^{*}_i$ and $\boldsymbol{\psi}^*$
		
		\ENDFOR
		
	\end{algorithmic}
\end{algorithm}

\subsection{CS Clustering-Based Energy Learning}
\label{sec:CSC}

Due to the learned model exchanges among large number of CSs, the use of CS-based DFEL method may slow down the learning process. Moreover, the CS-based DFEL algorithm may provide biased energy demand prediction if we do not take the important feature classfication into account when learning the dataset. This may occur if we join disparated features and defined labels into one dataset. To address these issues, we can classify all CSs in the EV network into $K$ clusters of CSs prior to the learning process. This aims to obtain better prediction accuracy and learning speed (as illustrated in Fig.~\ref{fig:learning_method}(b)). In particular, we first determine the clustering decision based on the public location information sharing of the CSs including the latitude and longitude. Then, we use the constrained K-means algorithm~\cite{Bradley:2000} to implement the clustering scheme. As such, we revise the constrained K-means algorithm to produce even distribution of the CSs in each cluster under the thresholds of minimum and maximum cluster sizes. In this way, we can guarantee that the learning process in each cluster achieves the fairness with respect to their deployment locations. 

Suppose that $\mathbf{G}$ is the dataset which corresponds to CS IDs and their locations, i.e., $g_i, \forall i \in \mathcal{I}$, and $\mathcal{K} = \{1,\ldots,k,\ldots,K\}$ is the set of clusters. We can find cluster centers, i.e., ${\hat g}_k,\forall k \in \mathcal{K}$, in such a way that we minimize the overall squared distance between each $g_i$ and its closest cluster center ${\hat g}_k$ as follows:
\begin{equation}
\label{eqn3j}
\begin{aligned}
\min_{\{\boldsymbol{\varphi}, \mathbf{q}\}}\sum_{i \in \mathcal{I}}\sum_{k \in \mathcal{K}}\varphi^k_i (g_i - {\hat g}_k)^2,
\end{aligned}
\end{equation}
\begin{eqnarray}
\text{s.t.} \quad
&\varphi_{\emph{\mbox{low}}}^k \leq \underset{i \in \mathcal{I}}{\sum}\varphi^k_i \leq \varphi_{\emph{\mbox{high}}}^k, \forall k \in \mathcal{K}, \label{eqn3j1} \\
&\underset{k \in \mathcal{K}}{\sum}\varphi^k_i = 1, \forall i \in \mathcal{I}, \label{eqn3j2} \\
&\varphi^k_i \in \{0,1\}, \forall i \in \mathcal{I}, \forall k \in \mathcal{K}, \label{eqn3j3} \\
&\varphi_{\emph{\mbox{low}}}^k \geq 0, \varphi_{\emph{\mbox{high}}}^k \geq 0,  \forall k \in \mathcal{K}, \label{eqn3j4}
\end{eqnarray}
where $\varphi^k_i$ is a binary variable that specifies the CS-$i$'s location from the cluster center ${\hat g}_k$. Specifically, if the CS-$i$'s location is the closest to the cluster center ${\hat g}_k$, then $\varphi^k_i=1$, otherwise $\varphi^k_i=0$. The constraints (\ref{eqn3j1}) ensure that the pre-defined minimum $\varphi_{\emph{\mbox{low}}}^k$ and maximum $\varphi_{\emph{\mbox{low}}}^k$ thresholds bound each cluster size. Moreover, the constraints (\ref{eqn3j2}) imply that the location of each CS is classified into one unique cluster only. 

At each iteration $\iota$, we update the cluster center ${\hat g}_k^{(\iota)}$ to achieve the optimal cluster solution. Particularly, if we have $\varphi_{\emph{\mbox{low}}}^k \leq \underset{i \in \mathcal{I}}{\sum}\varphi^k_i \leq \varphi_{\emph{\mbox{high}}}^k$ in the cluster-$k$, then ${\hat g}_k^{(\iota+1)}=\frac{\underset{i \in \mathcal{I}}{\sum}\varphi^{k,(\iota)}_{i}g_i}{\underset{i \in \mathcal{I}}{\sum}\varphi^{k,(\iota)}_{i}}$, and ${\hat g}_k^{(\iota+1)}={\hat g}_k^{(\iota)}$, otherwise.
The process completes when ${\hat g}_k^{(\iota+1)}={\hat g}_k^{(\iota)}, \forall k \in \mathcal{K}$. As a result, we can generate $\mathcal{I}_k$ as the optimal set of the CSs for each cluster-$k$. Upon completing the clustering process, we can execute the CS-based DFEL method to predict the CSs' energy demands in each cluster separately. In Algorithm~\ref{ConsKMeans}, we show the CS clustering-based DFEL algorithm utilizing revised constrained K-means optimization. 

\begin{algorithm}[t]
	\caption{CS Clustering-Based DFEL Algorithm} \label{ConsKMeans}
	
	\begin{algorithmic}[1] 
		
		\STATE Set $\mathbf{G}$ containing $g_i, \forall i \in \mathcal{I}$
		
		\STATE Determine $K$ and randomize ${\hat g}_k^{(\iota)},\forall k \in \mathcal{K}$ for $\iota=0,1$ with ${\hat g}_k^{(1)} \neq {\hat g}_k^{(0)}$ 
		
		\WHILE{${\hat g}_k^{(\iota+1)} \neq {\hat g}_k^{(\iota)}, \forall k \in \mathcal{K}$}
		
		\STATE Execute revised constrained K-Means in Eqs.~(\ref{eqn3j})-(\ref{eqn3j4})
		
		\STATE Update $\iota \leftarrow \iota+1$
		
		\FOR{$\forall k \in \mathcal{K}$}
		
		\IF{$\varphi_{low}^k \leq \underset{i \in \mathcal{I}}{\sum}\varphi^k_i \leq \varphi_{high}^k$}
		
		\STATE ${\hat g}_k^{(\iota+1)}=\frac{\underset{i \in \mathcal{I}}{\sum}\varphi^{k,(\iota)}_{i}g_i}{\underset{i \in \mathcal{I}}{\sum}\varphi^{k,(\iota)}_{i}}$
		
		\ELSE
		
		\STATE ${\hat g}_k^{(\iota+1)}={\hat g}_k^{(\iota)}$
		
		\ENDIF
		
		\ENDFOR
		
		\ENDWHILE
		
		\FOR {$\forall k \in \mathcal{K}$}
		
		\STATE Find the optimal set of CSs $\mathcal{I}_k$ in the cluster-$k$
		
		\STATE Implement CS-based DFEL approach using Algorithm~1 in the cluster-$k$
		
		\ENDFOR
		
	\end{algorithmic}
\end{algorithm}

\section{Multi-Principal One-Agent Contract-Based Problem Transformation}\label{sec:PFT3}

Based on the predicted energy demand for all CSs from the federated energy learning process, each CS can offer an initial contract, i.e., the predicted energy demand and offered payment, to the SGP for the energy contract iterative process. To this end, we need to simplify the optimization problem $(\mathbf{P}_3)$ by reducing the number of IR constraints in Eq.~(\ref{eqn:for9b}) and IC constraints in Eq.~(\ref{eqn:for9c}) due to the complexity (especially when the number of possible types grows significantly in the iterative process). Based on the following Lemma~\ref{lemma0a}, the computational complexity of $(\mathbf{P}_3)$ follows $O(\phi_{tot}^2)$, where $\phi_{tot}$ is the total number of the SGP's possible types.
\begin{lemma}
	\label{lemma0a}
	The problem $(\mathbf{P}_3)$ has computational complexity $O(\phi_{tot}^2)$. 
\end{lemma}
\begin{proof}
	See Appendix~\ref{appx:lemma0a}.
\end{proof}
To solve the problem faster, we can reduce the computational complexity into $O(\phi_{tot})$ through transforming the IR and IC constraints as explained in the following section. To do so, we first demonstrate that when the SGP's type is higher, CSs will request larger amount of energy and offer higher payments to the SGP. This condition can be formally written in Lemma~\ref{lemma1}.
\begin{lemma}
	\label{lemma1}
	Let $(\boldsymbol{\rho},\boldsymbol{\xi})$ denote any feasible contract from CSs to the SGP such that if $\phi > \phi^*$, then $\boldsymbol{\rho}(\phi) > \boldsymbol{\rho}(\phi^*)$, and if $\phi = \phi^*$, then $\boldsymbol{\rho}(\phi) = \boldsymbol{\rho}(\phi^*)$, where $\phi, \phi^* \in \Phi$.
\end{lemma}

\begin{proof}
	See Appendix~\ref{appx:lemma1}.
\end{proof}

Based on Lemma~\ref{lemma1}, we observe that the utility function of the SGP follows a monotonic increasing function of $\phi$, i.e., if $\phi > \phi^*$, then $\boldsymbol{\rho}(\phi) > \boldsymbol{\rho}(\phi^*)$ and $\phi G(\boldsymbol{\hat \pi},\boldsymbol{\rho}(\phi)) - C(\boldsymbol{\hat \pi},\boldsymbol{\xi}(\phi)) > \phi^* G(\boldsymbol{\hat \pi},\boldsymbol{\rho}(\phi^*)) - C(\boldsymbol{\hat \pi},\boldsymbol{\xi}(\phi^*))$. Consequently, we can lessen the number of IR constraints by using the $\phi_{min}$. Particularly, by utilizing the IC constraints, we obtain
\begin{equation}
\label{eqn:for12}
\begin{aligned}
\phi G(\boldsymbol{\hat \pi},\boldsymbol{\rho}(\phi)) - C(\boldsymbol{\hat \pi},\boldsymbol{\xi}(\phi)) \geq \\ \phi G(\boldsymbol{\hat \pi},\boldsymbol{\rho}(\phi_{min})) - C(\boldsymbol{\hat \pi},\boldsymbol{\xi}(\phi_{min}))
\geq \\ \phi_{min} G(\boldsymbol{\hat \pi},\boldsymbol{\rho}(\phi_{min})) - C(\boldsymbol{\hat \pi},\boldsymbol{\xi}(\phi_{min})) \geq 0.
\end{aligned}
\end{equation}
In other words, the IR constraints for other $\phi$, where $\phi > \phi_{min}$, will hold as long as the IR constraint for $\phi_{min}$ is satisfied. Alternatively, we can transform the IR constraints in Eq.~(\ref{eqn:for9b}) into
\begin{equation}
\label{eqn:for13}
\begin{aligned}
\phi_{min} G(\boldsymbol{\hat \pi},\boldsymbol{\rho}(\phi_{min})) - C(\boldsymbol{\hat \pi},\boldsymbol{\xi}(\phi_{min})) \geq 0.
\end{aligned}
\end{equation}

Similar to the IR constraint simplification, we can reduce the number of IC constraints by transforming them using the following conditions stated in Lemma~\ref{lemma2}.
\begin{lemma}
	\label{lemma2}
	The IC constraints in Eq.~(\ref{eqn:for9c}) of $(\mathbf{P}_3)$ are equivalent to the following monotonicity and local incentive compatibility conditions, i.e.,
	\begin{equation}
	\label{eqn:for140}
	\begin{aligned}
	\frac{d \boldsymbol{\rho}(\phi)}{d\phi} \geq 0, \forall \phi \in \Phi,
	\end{aligned}
	\end{equation}
	and
	\begin{equation}
	\label{eqn:for14}
	\begin{aligned}
	\phi \frac{dG(\boldsymbol{\hat \pi},\boldsymbol{\rho}(\phi))}{d\phi} - \frac{dC(\boldsymbol{\hat \pi},\boldsymbol{\xi}(\phi))}{d\phi} \geq 0, \forall \phi \in \Phi,
	\end{aligned}
	\end{equation}
	respectively.
\end{lemma}
\begin{proof}
	See Appendix~\ref{appx:lemma2}.
\end{proof}
The conditions in Eq.~(\ref{eqn:for140}) imply that a higher type of the SGP will require higher offered payments from CSs as described in Lemma~\ref{lemma1}. Furthermore, the conditions in Eq.~(\ref{eqn:for14}) indicate that if for each type $\phi$, the IC constraint regarding the type that is lower than $\phi$ holds, then all other IC constraints are also satisfied as long as the conditions in Eq.~(\ref{eqn:for140}) hold. 

Based on the aforementioned constraint transformation in Eqs.~(\ref{eqn:for13})-(\ref{eqn:for14}), we can rewrite the optimization problem $(\mathbf{P}_3)$ into
\begin{equation}
\label{eqn:for9rev}
\begin{aligned}
(\mathbf{P}_4) \phantom{10} & \underset{\boldsymbol{\rho}(\phi),\boldsymbol{\xi}(\phi)}{\text{max}} \phantom{5} U_{i}(\boldsymbol{\rho}(\phi), \boldsymbol{\xi}(\phi)), \forall i \in \mathcal{I},
\end{aligned}
\end{equation}
\begin{eqnarray}
\text{ s.t. } \quad &\text{
	(\ref{eqn:for9a}) and}, \nonumber \\ 
&\phi_{min} G(\boldsymbol{\hat \pi},\boldsymbol{\rho}(\phi_{min})) - C(\boldsymbol{\hat \pi},\boldsymbol{\xi}(\phi_{min})) \geq 0, \label{eqn:for9revb} \\
&\phi \frac{dG(\boldsymbol{\hat \pi},\boldsymbol{\rho}(\phi))}{d\phi} - \frac{dC(\boldsymbol{\hat \pi},\boldsymbol{\xi}(\phi))}{d\phi} \geq 0, \forall \phi \in \Phi, \label{eqn:for9revc} \\
&\frac{d \boldsymbol{\rho}(\phi)}{d\phi} \geq 0, \forall \phi \in \Phi. \label{eqn:for9revd}
\end{eqnarray}
Then, we can transfrom the left side of the constraints (\ref{eqn:for9revc}) as below
\begin{equation}
\label{eqn:for15e}
\begin{aligned}
&\phi \frac{dG(\boldsymbol{\hat \pi},\boldsymbol{\rho}(\phi))}{d\phi} - \frac{dC(\boldsymbol{\hat \pi},\boldsymbol{\xi}(\phi))}{d\phi} \\
&=\frac{\phi}{1 + \overset{I}{\underset{i=1}{\sum}} {\hat \pi}_i\rho_i(\phi)}\overset{I}{\underset{i=1}{\sum}} {\hat \pi}_i\frac{d\rho_i(\phi)}{d\phi} - \zeta\overset{I}{\underset{i=1}{\sum}} {\hat \pi}_i\frac{d\xi_i(\phi)}{d\phi} \\
&= \phi \overset{I}{\underset{i=1}{\sum}} {\hat \pi}_i \frac{d\rho_i(\phi)}{d\phi} - \Bigg(1 + \overset{I}{\underset{i=1}{\sum}} {\hat \pi}_i\rho_i(\phi)\Bigg) \Bigg(\zeta\overset{I}{\underset{i=1}{\sum}} {\hat \pi}_i \frac{d\xi_i(\phi)}{d\phi}\Bigg).
\end{aligned}
\end{equation}
Since the SGP has discrete number of possible types, i.e., $1, 2, \ldots, \phi_{max}$, the gap between two consecutive types is 1, i.e., $d\phi = 1$. Thus, we can modify $\frac{d\rho_i(\phi)}{d\phi}$ and $\frac{d\xi_i(\phi)}{d\phi}$ by
\begin{equation}
\label{eqn:for15e2}
\begin{aligned}
\frac{d\rho_i(\phi)}{d\phi} &= \frac{\rho_i(\phi) - \rho_i(\phi - d\phi)}{d\phi} \\
&= \rho_i(\phi) - \rho_i(\phi - 1),
\end{aligned}
\end{equation}
and
\begin{equation}
\label{eqn:for15e3}
\begin{aligned}
\frac{d\xi_i(\phi)}{d\phi} &= \frac{\xi_i(\phi) - \xi_i(\phi - d\phi)}{d\phi} \\
&= \xi_i(\phi) - \xi_i(\phi - 1),
\end{aligned}
\end{equation}
respectively. Then, from Eq.~(\ref{eqn:for15e}), we have
\begin{equation}
\label{eqn:for15f}
\begin{aligned}
&\phi \frac{dG(\boldsymbol{\hat \pi},\boldsymbol{\rho}(\phi))}{d\phi} - \frac{dC(\boldsymbol{\hat \pi},\boldsymbol{\xi}(\phi))}{d\phi} \\
&= \phi \overset{I}{\underset{i=1}{\sum}} {\hat \pi}_i\Big(\rho_i(\phi) - \rho_i(\phi - 1)\Big) - \\ &\Bigg(1 + \overset{I}{\underset{i=1}{\sum}} {\hat \pi}_i\rho_i(\phi)\Bigg)\Bigg(\zeta\overset{I}{\underset{i=1}{\sum}} {\hat \pi}_i\Big(\xi_i(\phi) - \xi_i(\phi - 1)\Big)\Bigg). \\
\end{aligned}
\end{equation}
As a result, the simplified version of $(\mathbf{P}_4)$ is
\begin{equation}
\label{eqn:for9rev2}
\begin{aligned}
(\mathbf{P}_5) \phantom{10} & \underset{\boldsymbol{\rho}(\phi),\boldsymbol{\xi}(\phi)}{\text{max}} \phantom{5} {U}_{i}(\boldsymbol{\rho}(\phi), \boldsymbol{\xi}(\phi)), \forall i \in \mathcal{I},
\end{aligned}
\end{equation}
\begin{eqnarray}
\text{ s.t. } \quad &\text{
	(\ref{eqn:for9a}), 
	(\ref{eqn:for9revb}) and}, \nonumber \\ 
&\phi \overset{I}{\underset{i=1}{\sum}} {\hat \pi}_i\Big(\rho_i(\phi) - \rho_i(\phi - 1)\Big) - \nonumber \\  
&\Bigg(1 + \overset{I}{\underset{i=1}{\sum}} {\hat \pi}_i\rho_i(\phi)\Bigg)\Bigg(\zeta\overset{I}{\underset{i=1}{\sum}} {\hat \pi}_i\Big(\xi_i(\phi) - \xi_i(\phi - 1)\Big)\Bigg) \nonumber \\ &\geq 0, \forall \phi \in \Phi, \label{eqn:for9revc2} \\
&\rho_i(\phi) - \rho_i(\phi - 1)  \geq 0, \forall i \in \mathcal{I}, \forall \phi \in \Phi. \label{eqn:for9revd2}
\end{eqnarray}
Finally, we demonstrate that the problem $(\mathbf{P}_5)$ has computational complexity $O(\phi_{tot})$ as stated in the following Lemma~\ref{lemma0b}.
\begin{lemma}
	\label{lemma0b}
	The problem $(\mathbf{P}_5)$ has computational complexity $O(\phi_{tot})$. 
\end{lemma}
\begin{proof}
	See Appendix~\ref{appx:lemma0b}.
\end{proof}

\section{Non-Collaborative Energy Contract Solution}
\label{sec:PS}

\subsection{Energy Contract Iterative Algorithm}

To find the optimal contracts from $(\mathbf{P}_5)$, we propose an iterative algorithm as shown in Algorithm~1. In particular, we first find the optimal values of $\boldsymbol{\hat \pi}$  maximize the objective function of nonlinear programming problem $(\mathbf{P}_1)$. 
Given $\boldsymbol{\hat \pi}$ and other CSs' current contracts remain pre-defined~\cite{Voorneveld:2000, Zhu:2015} at the SGP, we execute the iterative algorithm. Using this algorithm, the SGP can update the possible contract of each CS, which maximizes the CS's expected utility for each iteration. Specifically, the current contract of each CS-$i$ can be applied if its current expected utility is higher than the previous one when the previous contract of the CS is used. Otherwise, the previous contract will be utilized. The algorithm will terminate when the expected utilities of all CSs reach the optimality tolerance $\kappa$ (i.e., no further improvement in the expected utility), and thus the algorithm converges where the equilibrium contract solution is achieved. In this way, the SGP can provide fair competition among the participating CSs based on their offered energy contracts.

\begin{algorithm}[]
	\caption{Energy Contract Iterative Algorithm} \label{ICEA}
	
	\begin{algorithmic}[1] 
		\STATE The SGP notifies the current price for a unit energy (MWh) to all CSs
		
		\STATE Initialize $\kappa$ and $\theta = 0$
		
		\STATE Each CS-$i$ offers initial contract $\Big(\rho_i^{(\theta)}(\phi),\xi_i^{(\theta)}(\phi)\Big), \forall \phi \in \Phi$, to the SGP
		
		
		
		
		
		\REPEAT
		
		\STATE Find $\boldsymbol{\hat \pi}^{(\theta)}$ which maximize $(\mathbf{P}_1)$ given $\Big(\boldsymbol{\rho}^{(\theta)}(\phi), \boldsymbol{\xi}^{(\theta)}(\phi)\Big)$ and the SGP's type $\phi$
		
		\FOR{$\forall i \in \mathcal{I}$}
		
		\STATE Obtain a new contract $\Big(\rho_i^{(\theta+1)}(\phi),\xi_i^{(\theta+1)}(\phi)\Big), \forall \phi \in \Phi$, which maximizes $(\mathbf{P}_5)$ using $\boldsymbol{\hat \pi}^{(\theta)}$ and $\Big(\boldsymbol{\rho}_{-i}^{(\theta)}(\phi),\boldsymbol{\xi}_{-i}^{(\theta)}(\phi)\Big), \forall \phi \in \Phi$
		
		
		\IF{$\bigg[{ U}_{i}\Big(\rho^{(\theta+1)}_i(\phi),\xi^{(\theta+1)}_i(\phi),\boldsymbol{\rho}^{(\theta)}_{-i}(\phi), \boldsymbol{\xi}^{(\theta)}_{-i}(\phi)\Big) - {U}_{i}\Big(\rho^{(\theta)}_i(\phi),\xi^{(\theta)}_i(\phi),\boldsymbol{\rho}^{(\theta)}_{-i}(\phi), \boldsymbol{\xi}^{(\theta)}_{-i}(\phi)\Big)\bigg] > \kappa$}
		
		\STATE Set $\Big(\rho_i^{(\theta)}(\phi),\xi_i^{(\theta)}(\phi)\Big) = \Big(\rho_i^{(\theta+1)}(\phi),\xi_i^{(\theta+1)}(\phi)\Big)$,  $\forall \phi \in \Phi$  
		
		\ENDIF
		
		\ENDFOR
		
		\UNTIL {${U}_{i}\Big(\rho^{(\theta)}_i(\phi),\xi^{(\theta)}_i(\phi),\boldsymbol{\rho}^{(\theta)}_{-i}(\phi), \boldsymbol{\xi}^{(\theta)}_{-i}(\phi)\Big), \forall i \in \mathcal{I}$, do not change anymore}
		
		\STATE Attain $\Big(\boldsymbol{\hat \rho}(\phi), \boldsymbol{\hat \xi}(\phi)\Big)$, $\forall \phi \in \Phi$
		
	\end{algorithmic}
\end{algorithm}

\subsection{Convergence and Equilibrium Contract Analysis}

In this section, we investigate the convergence and equilibrium contract solution for the proposed non-collaborative MPOA contract problem. Specifically, we define the communication between the SGP and CSs as a two-stage game~\cite{Dixit:1997} to find the equilibrium. In the first stage, each CS generates a contract and notices that other CSs non-collaboratively choose their own contracts simultaneously at the same time. These selected contracts are then sent to the SGP for the second stage process. As such, the SGP finds the optimal energy proportions of the CSs based on its type to maximize its own utility. Since each CS does not have any contract information from other CSs due the privacy concern, the SGP can help the CSs to process the second stage locally. Then, we only need to show that the Algorithm~\ref{ICEA} converges to the equilibrium contract solution. This can be executed through observing the best responses~\cite{Voorneveld:2000} containing the best contracts from all the CSs at each iteration. Specifically, the best response of CS-$i$ given $\Big(\boldsymbol{\rho}^{(\theta)}_{-i}(\phi), \boldsymbol{\xi}^{(\theta)}_{-i}(\phi)\Big)$ at iteration $\theta+1$ can be defined by
\begin{equation}
\label{eqn:for11a}
\begin{aligned}
\Gamma^{(\theta+1)}_i\Big(\boldsymbol{\rho}^{(\theta)}_{-i}(\phi), \boldsymbol{\xi}^{(\theta)}_{-i}(\phi)\Big) = \\ \underset{\{\rho_i(\phi), \xi_i(\phi)\} \in \mathbb{C}_i}{\arg \max} {U}_i\Big(\rho_i^{(\theta+1)}(\phi), \xi_i^{(\theta+1)}(\phi),\boldsymbol{\rho}^{(\theta)}_{-i}(\phi), \boldsymbol{\xi}^{(\theta)}_{-i}(\phi)\Big),
\end{aligned}
\end{equation}
where $\mathbb{C}_i$ is the non-empty contract space~\cite{Fraysse:1993} for CS-$i$ and $\mathbb{C} = \underset{i \in \mathcal{I}}{\prod}\mathbb{C}_i$. Based on the Algorithm~\ref{ICEA}, the current contract of CS-$i$ can be updated to $\Big(\rho^{(\theta+1)}_i(\phi),\xi^{(\theta+1)}_i(\phi)\Big) \in \Gamma^{(\theta+1)}_i$ if the following condition holds
\begin{equation}
\label{eqn:for11c}
\begin{aligned}
{U}_{i}\Big(\rho^{(\theta+1)}_i(\phi),\xi^{(\theta+1)}_i(\phi),\boldsymbol{\rho}^{(\theta)}_{-i}(\phi), \boldsymbol{\xi}^{(\theta)}_{-i}(\phi)\Big) - \\ {U}_{i}\Big(\rho^{(\theta)}_i(\phi),\xi^{(\theta)}_i(\phi),\boldsymbol{\rho}^{(\theta)}_{-i}(\phi), \boldsymbol{\xi}^{(\theta)}_{-i}(\phi)\Big) > \kappa.
\end{aligned}
\end{equation} 
The process continues until the algorithm converges for all CSs as described in Theorem~\ref{theorem_equi1}. 
\begin{theorem}\label{theorem_equi1}
	The best response iterative process in Algorithm~\ref{ICEA} converges under the optimality tolerance $\kappa$.
\end{theorem}
\begin{proof}
	See Appendix~\ref{appx:theorem1}.
\end{proof}

To guarantee that the algorithm converges to the equilibrium contract solution $\Big(\boldsymbol{\hat \rho}(\phi), \boldsymbol{\hat \xi}(\phi)\Big), \forall \phi \in \Phi$, we first observe that the equilibrium contract solution exists through identifying a fixed point in a set-valued function $\Gamma$, $\Gamma: \mathbb{C} \rightarrow 2^\mathbb{C}$, such that 
\begin{equation}
\label{eqn:for11d}
\begin{aligned}
\Gamma = \bigg[\Gamma_i\Big(\boldsymbol{\rho}_{-i}(\phi), \boldsymbol{\xi}_{-i}(\phi)\Big),\Gamma_{-i}\Big(\boldsymbol{\rho}_i(\phi), \boldsymbol{\xi}_i(\phi)\Big)\bigg].
\end{aligned}
\end{equation} 
The existence of this fixed point is equivalent to the equilibrium solution~\cite{Bernheim:1986, Fraysse:1993}. Then, we can show that the Algorithm 3 converges to the equilibrium contract solution, i.e., all the CSs obtain the maximum expected utilities where the $\Big(\boldsymbol{\hat \rho}(\phi), \boldsymbol{\hat \xi}(\phi)\Big)$ is found, as formally stated in Theorem~\ref{theorem_equi2}.
\begin{theorem}\label{theorem_equi2}
	The Algorithm 3 drives the best response iterative process to its equilibrium contract solution $\Big(\boldsymbol{\hat \rho}(\phi), \boldsymbol{\hat \xi}(\phi)\Big), \forall \phi \in \Phi$, through obtaining a fixed point $\Big(\boldsymbol{\rho}^*(\phi), \boldsymbol{\xi}^*(\phi)\Big), \forall \phi \in \Phi$, in $\Gamma$.
\end{theorem}
\begin{proof}
	See Appendix~\ref{appx:theorem2}.
\end{proof}

\section{Performance Evaluation} 
\label{sec:PE}

\subsection{Dataset Preparation and Evaluation Method for Energy Demand Prediction} 
\label{subsec:DP}

In the simulations, we utilize the actual dataset generated from CSs' transactions in Dundee city, the United Kingdom between 2017 and 2018~\cite{Dundee:2018} to show the efficiency of the proposed learning and economic approaches. Specifically, the dataset contains 65,601 transactions of charging EVs with the following information: CS unique identifier (i.e., 58 CSs), EV transaction identifier, EV transaction date, EV transaction time, and energy usage (in kWh). We divide the information into four learning features, i.e., the first four information, and one label, i.e., the energy usage. Additionally, we group CS identifier, EV transaction date, and EV transaction time to be categorical features. In this case, we transform the EV transaction date and EV transaction time into 7-day and 24-hour categories, respectively. Moreover, as seen in Fig.~\ref{fig:Clustering_Topology}(a), each CS has the location information for the clustering purposes.


To observe the prediction accuracy, we compute the prediction error using RMSE. This is because we account for the energy demand  prediction which is classified as a regression prediction model, i.e., when the output layer of DNN produces the non-discrete prediction results. Given $M$ number of transactions, the RMSE can be calculated by
\begin{equation}
\label{eqn3k}
\begin{aligned}
\mbox{RMSE} = \sqrt{\frac{1}{M}\sum_{m=1}^M (\xi_m - {\hat \xi}_m)^2},
\end{aligned}
\end{equation}
where $\xi_m$ and ${\hat \xi}_m$ are the real and predicted energy demand for transaction $m$. 

\subsection{Simulation Setup} 
\label{subsec:ES}

\begin{figure*}[!]
	\begin{center}
		$\begin{array}{cc} 
		\epsfxsize=3.3 in \epsffile{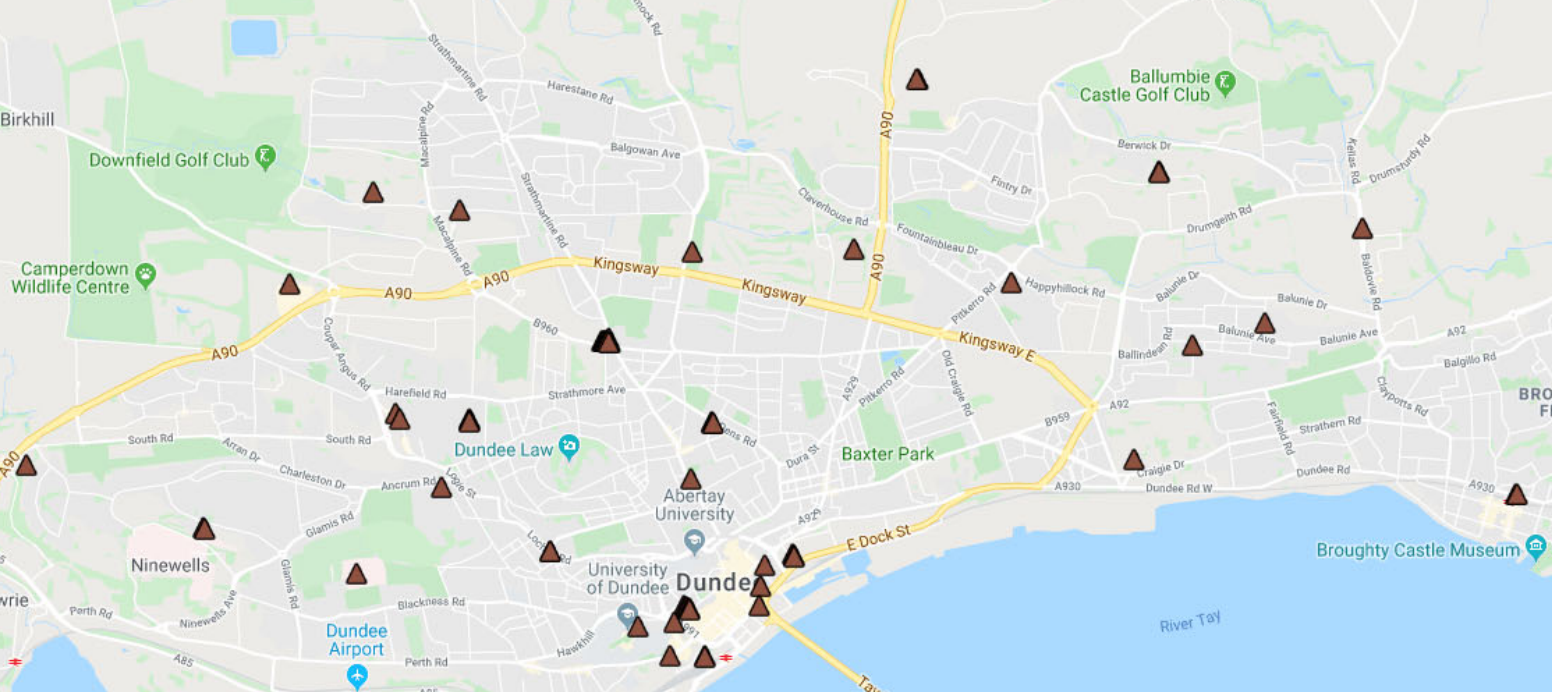} &
		\hspace*{0cm}
		\epsfxsize=3.38 in \epsffile{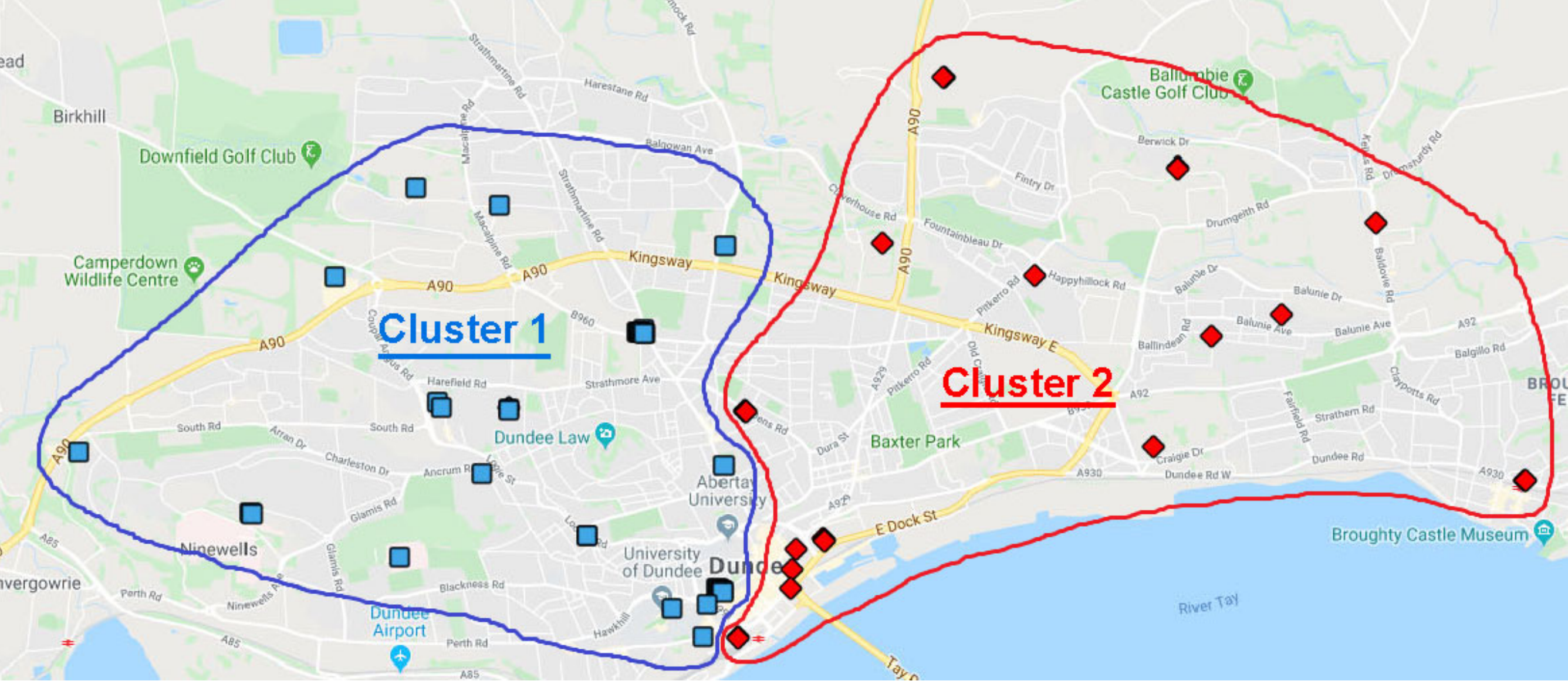} \\ [0.1cm]
		\text{\footnotesize (a) CSs without clustering} & \text{\footnotesize (b) CSs with clustering}  \\ [0.2cm]
		\end{array}$
		\caption{The distribution of CSs in Dundee city, the United
			Kingdom between 2017 and 2018~\cite{Dundee:2018}.}
		\label{fig:Clustering_Topology}
	\end{center}
\end{figure*}

We use \emph{TensorFlow} in a shared computing platform with Intel Xeon E5-2687W v2 3.4GHz 8 cores 32GB RAM to compare the performance of the proposed learning techniques with other centralized learning methods, i.e., a cloud server gathers actual datasets from all CSs and then performs the learning algorithms to predict the energy demand. These centralized methods include decision tree, random forest, support vector regressor, k-neighbors regressor, stochastic gradient descent regressor, multi-layer perceptron regressor, and cloud-based deep learning~\cite{Boutaba:2018}. We divide the dataset into training dataset with ratios 0.8, 0.7, 0.6, and 0.5, and testing dataset with the rest of the ratios. We also split all transactions into $I$ training subsets. Then, each CS-$i$ runs the training subsets as well as testing dataset for the next energy demand prediction. For the CS clustering-based DFEL method, we set $K=2$ and split all the CSs into 2 clusters as illustrated in Fig.~\ref{fig:Clustering_Topology}(b). For the learning process using DNN, we add three layers including two hidden layers with the same number of neurons, i.e., 64 neurons for each layer, which are followed by one dropout layer with a fraction rate 0.15. Moreover, we use the tanh activation function and Adam optimizer which starts from step size 0.01.

Then, we compare the performance of proposed economic model for energy transfer activity with the information-symmetry~\cite{Chen2:2019}, proportional-request, and non-prediction methods~\cite{Tang:2014}.
For the information-symmetry method, the CSs exactly know the type, i.e., the energy capacity, of the SGP to find the optimal contract policy. In this case, we use the information-symmetry method as the upper bound solution. For the proportional-request method, each CS obtains the proportional amount of energy according to its energy demand without using contract policy. Additionally, for the non-prediction method, the CSs request energy transfer from the SGP immediately once they receive energy demands from EVs (with fluctuated energy price unit and without considering contract policy).
In the contract mechanism, we consider one agent, i.e., the SGP, and 58 principals corresponding to 58 CSs. We set the price unit of energy transfer from the SGP at 200 monetary units (MU) for the proposed, information-symmetry, and proportional-request methods. We also set energy charging price of EVs at 220MU per MWh~\cite{electriccar} for all methods. For the proposed, information-symmetry, and proportional-request methods, the initial energy demands for all CSs are generated from the energy demand prediction of CS-based DFEL when 0.8 training set ratio is used. To show various results of the SGP, we consider 10 to 50 possible types with the same distribution of the types, i.e., $p(\phi)=\frac{1}{\phi_{max}}$. As such, each type $\phi$ corresponds to the energy capacity $\phi \frac{S_{max}}{\phi_{max}}$ MWh of the SGP, where $S_{max}$ is the maximum capacity of the SGP regardless the number of possible types and set at 500MW. Finally, we set $\zeta$ at 0.022.

\subsection{Prediction Accuracy and Communication Overhead Performance}
\label{subsec:SR1}

We show the comparisons between centralized and proposed machine learning methods for different ratios of training set in Table~\ref{tab:contents2}. First, we analyze the RMSE, i.e., prediction accuracy, of the testing set when we use 0.8 training set ratio. Specifically, we can observe that the CS-based DFEL with clustering can reduce the RMSE up to 24.63\%, respectively, compared with those of the centralized methods. The reason is that the CS clustering-based method can combine similar important features and/or labels in the same cluster, and thus improve the prediction accuracy~\cite{Li:2018}. In this way, we can minimize the biased prediction cost of the whole dataset by clustering the CSs based on their locations, which then produces the lower prediction error. Furthermore, for the CS-based DFEL without clustering, the performance of RMSE can achieve 23.51\% lower than those of the centralized learning methods. As such, the proposed learning method without clustering can obtain less than 2\% gap from the ones with clustering. The reason is that the CS-based DFEL can thoroughly learn the subset of the whole dataset individually at diverse workers and obtain the average prediction with lower error as well as less variance regarding the number of workers~\cite{Guo:1999}. We also observe that the CS-based DFEL with and without clustering still outperform all the centralized learning methods for other training set ratios, i.e., 0.7, 0.6, and 0.5. To be more specific, the CS-based DFEL with clustering has the lowest RMSE, i.e., the best prediction accuracy, for those training set ratios.

\begin{table}[!]
	\centering
	\caption{The comparison of testing RMSE for different learning methods.}
	\begin{center}
		\begin{tabular}{ |c|c|c|c|c| } 
			\hline
			\multirow{2}{*}{\bf Energy learning method} & \multicolumn{4}{|c|}{\bf Ratio of training set} \\
			\cline{2-5}
			& {\bf 0.8} & {\bf 0.7} & {\bf 0.6} & {\bf 0.5} \\
			\hline
			\hline
			K-neighbors regressor & 7.18 & 7.71 & 7.57 & 7.67 \\
			\hline 
			Multi-layer perceptron & \multirow{2}{*}{6.57} & \multirow{2}{*}{6.62} & \multirow{2}{*}{6.90} & \multirow{2}{*}{6.53} \\
			regressor & & & &\\
			\hline 
			Stochastic gradient & \multirow{2}{*}{6.55} & \multirow{2}{*}{6.57} & \multirow{2}{*}{6.54} & \multirow{2}{*}{6.54} \\ 
			descent regressor & & & &\\ 
			\hline
			Decision tree & 6.47 & 6.49 & 6.47 & 6.47 \\ 
			\hline
			Support vector regressor & 6.46 & 6.50 & 6.50 & 6.53 \\ 
			\hline
			Random forest & 6.35 & 6.66 & 6.88 & 6.80 \\ 
			\hline
			Cloud-based deep learning & 5.86 & 5.86 & 5.86 & 5.87 \\ 
			\hline
			{\bf CS-based DFEL} & {\bf 5.81} & {\bf 5.82} & {\bf 5.81} & {\bf 5.84} \\ 
			\hline
			{\bf  CS-based DFEL +} & \multirow{2}{*}{\bf 5.76} & \multirow{2}{*}{\bf 5.78} & \multirow{2}{*}{\bf 5.78} & \multirow{2}{*}{\bf 5.83}  \\
			{\bf  CS Clustering} & & & &\\
			\hline
		\end{tabular}
	\end{center}
	\label{tab:contents2}
\end{table}

\begin{figure}[!t]
	\begin{center}
		$\begin{array}{cc} 
		\epsfxsize=1.68 in \epsffile{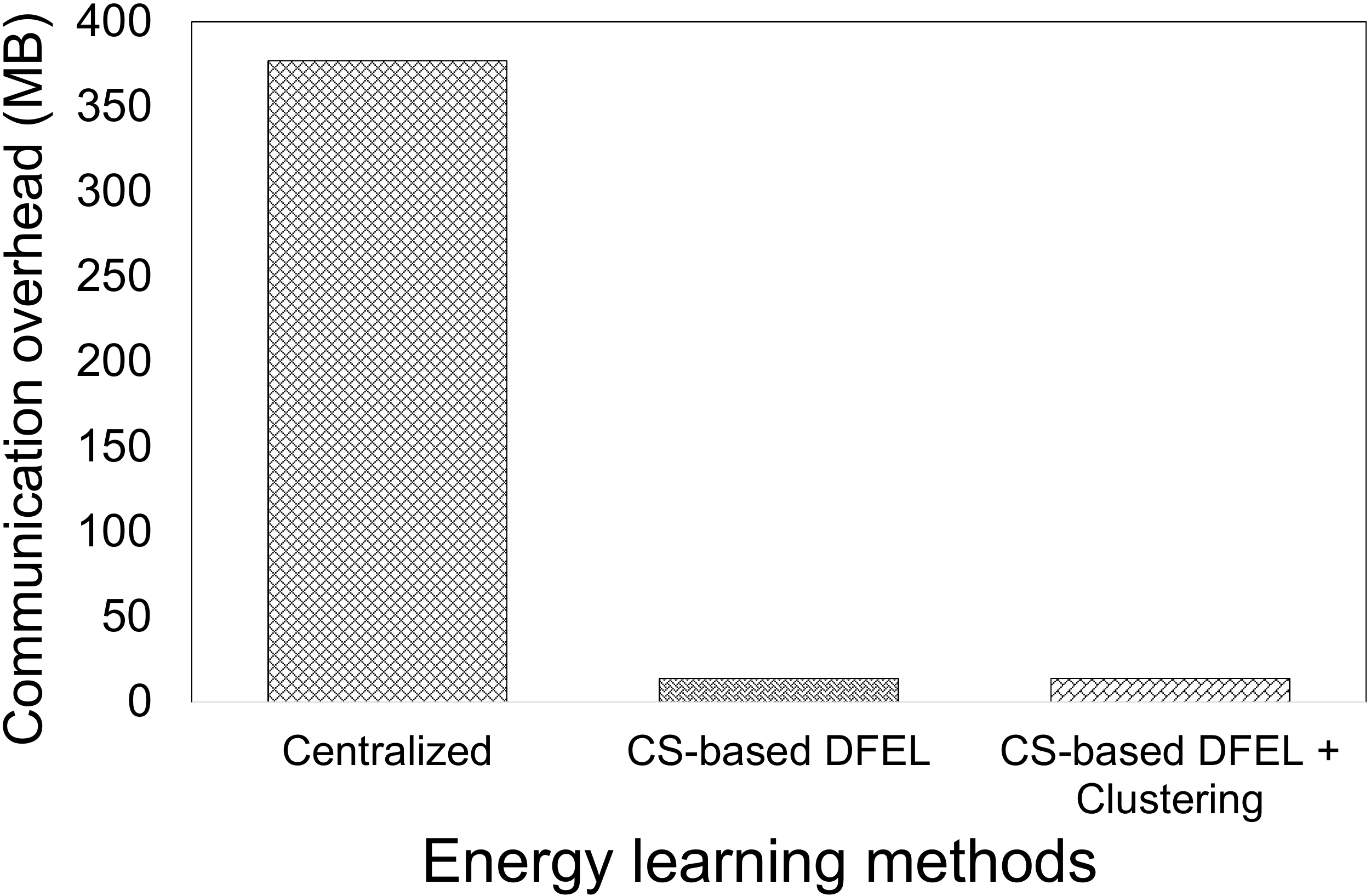} &
		\hspace*{-0.2cm}
		\epsfxsize=1.66 in \epsffile{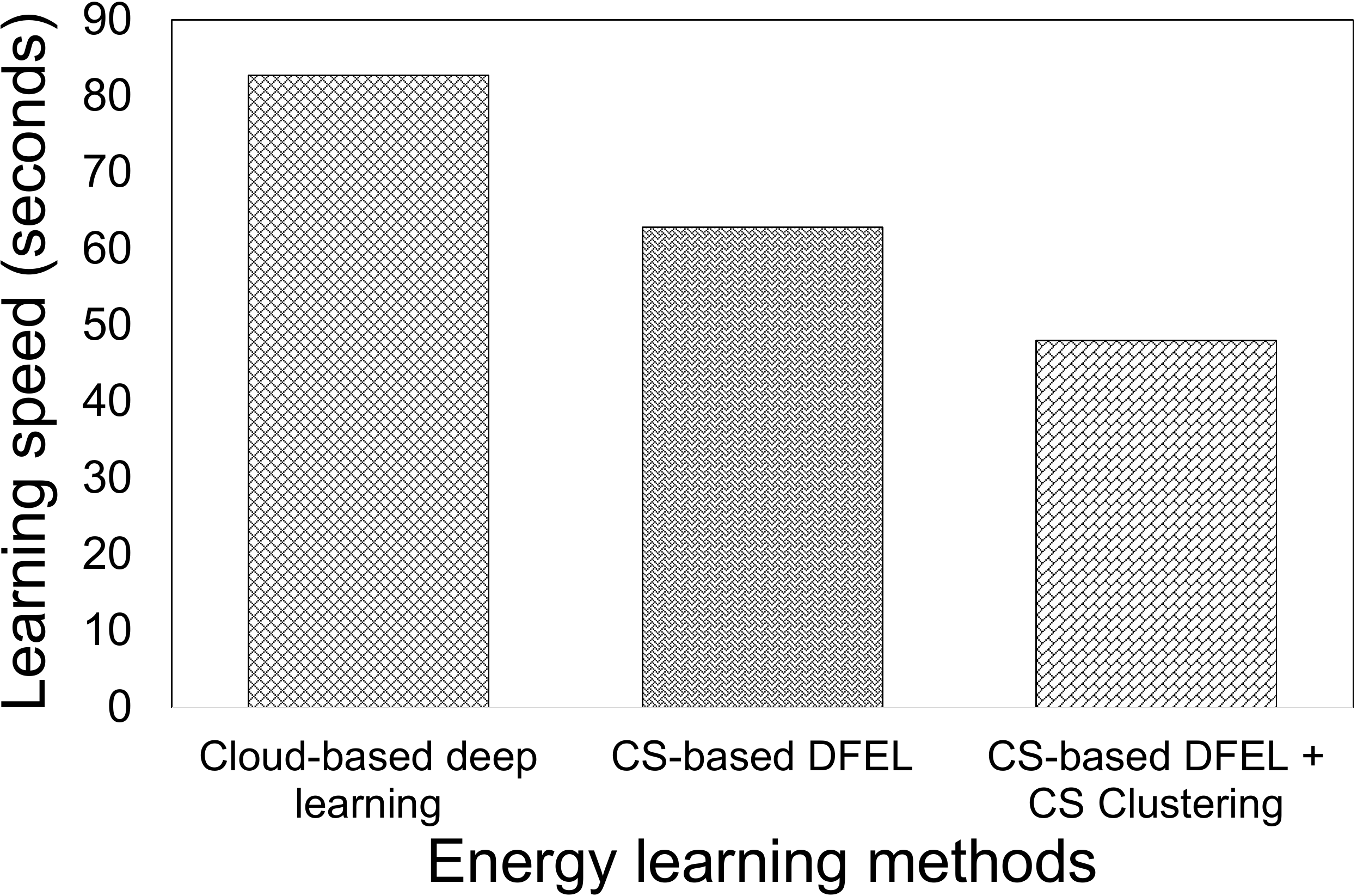} \\ [0.1cm]
		\text{\footnotesize (a) Communication overhead} & \text{\footnotesize (b) Learning speed}  \\ [0.2cm]
		\end{array}$
		\caption{The performance of communication overhead and learning speed for various energy learning methods.}
		\label{fig:comm_over}
	\end{center}
\end{figure}

\begin{figure*}[!t]
	\begin{center}
		$\begin{array}{ccc} 
		\epsfxsize=2.24 in \epsffile{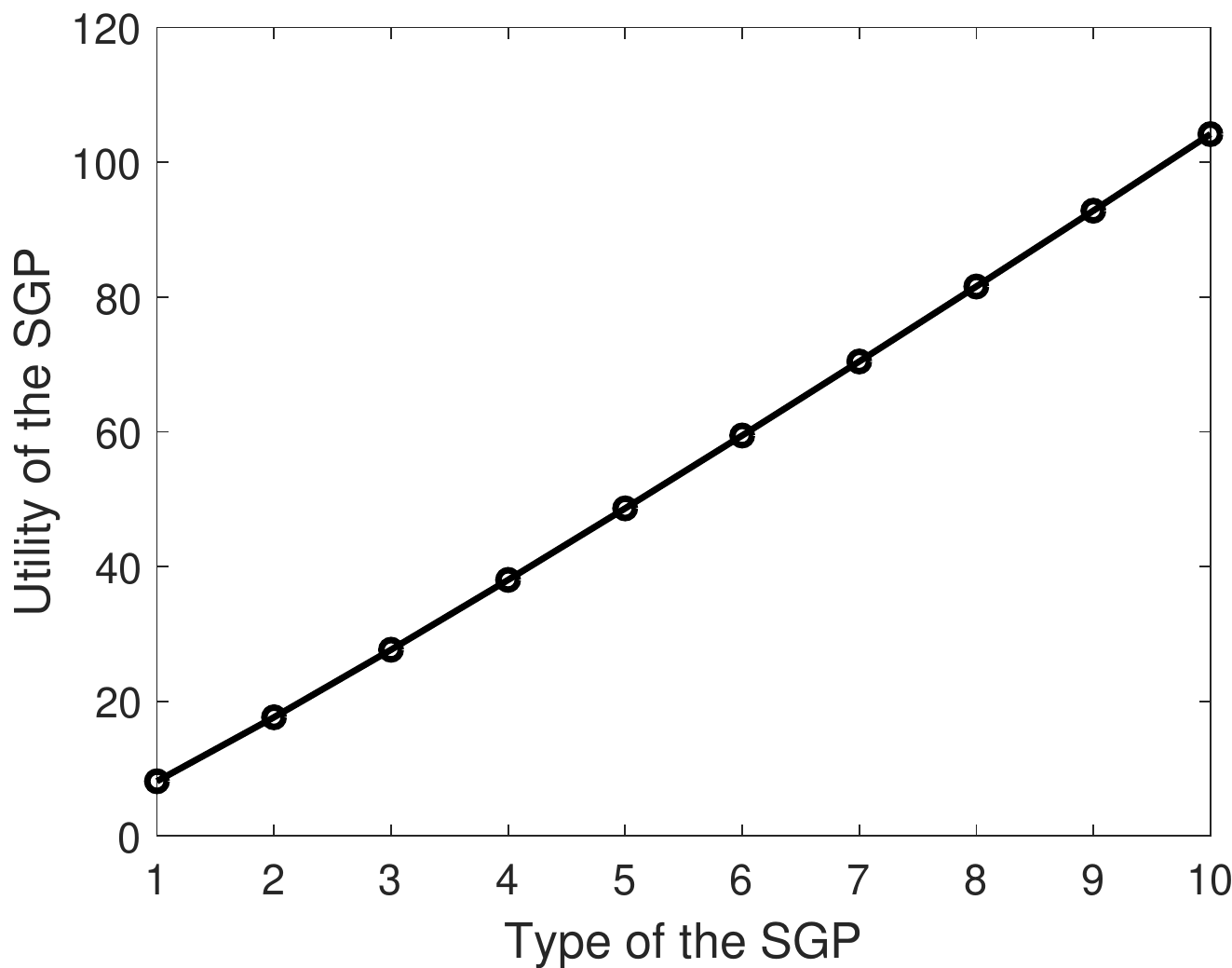} & 
		\hspace*{-.1cm}
		\epsfxsize=2.2 in \epsffile{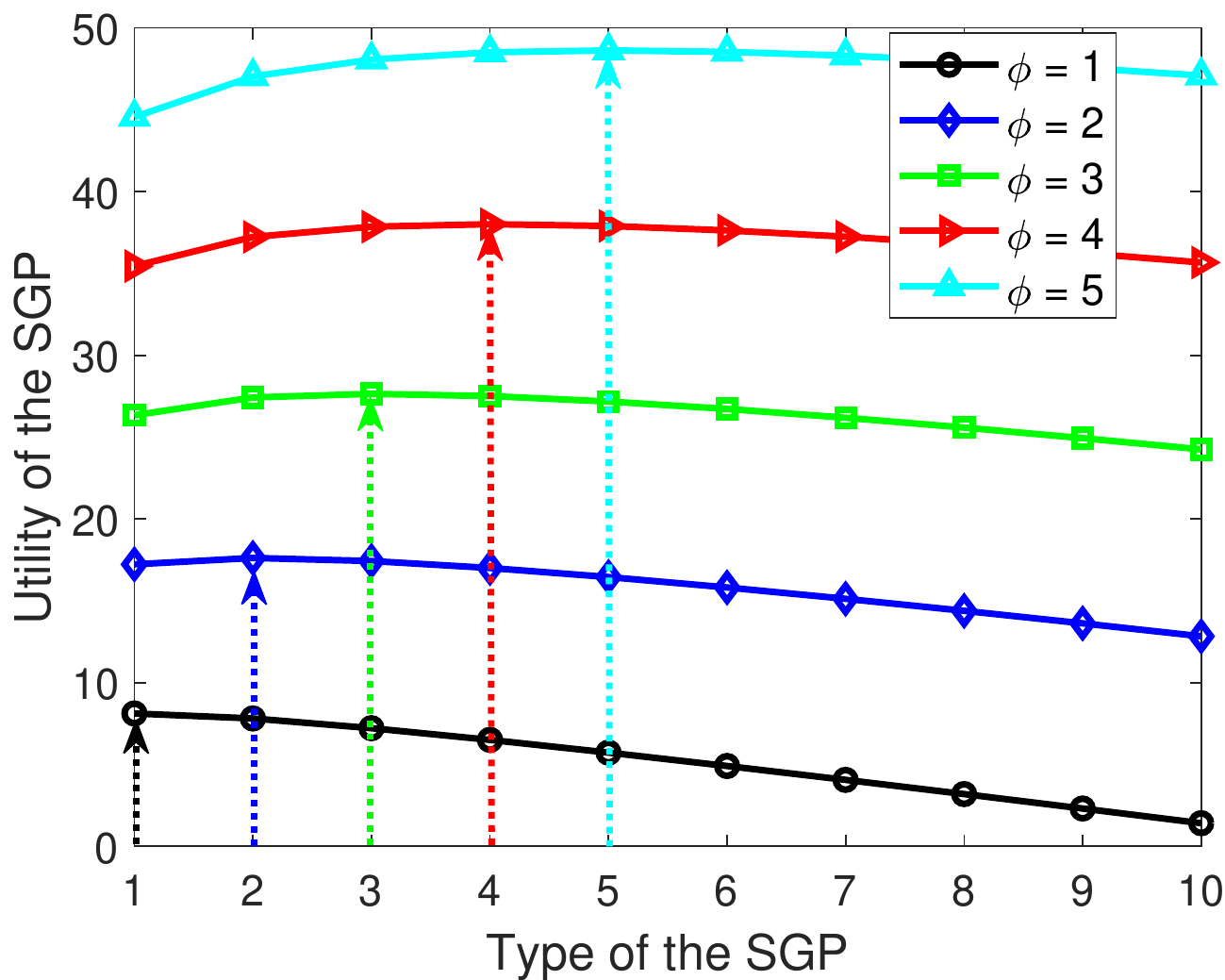} & 
		\hspace*{-.1cm}
		\epsfxsize=2.23 in \epsffile{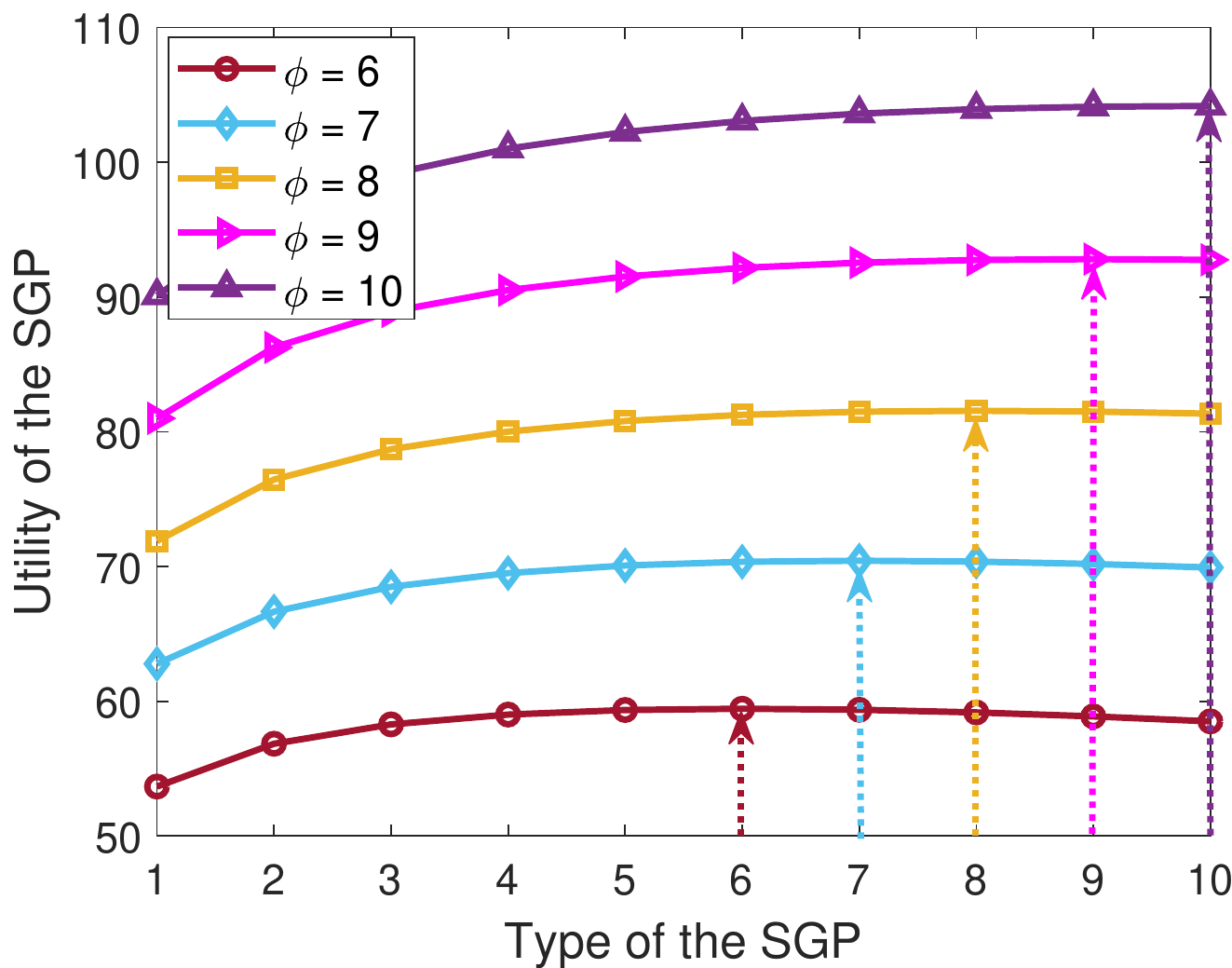} \\ [0.1cm]
		\text{\footnotesize (a) IR constraint} & \text{\footnotesize (b) IC constraints for $\phi = 1$ to $\phi = 5$} & \text{\footnotesize (c) IC constraints for $\phi = 6$ to $\phi = 10$} \\ [0.2cm]
		\end{array}$
		\caption{The observation of the SGP's IR and IC constraints for the proposed economic model.}
		\label{fig:util_SGP_IC}
	\end{center}
\end{figure*} 

Different from the proposed learning methods, the first six centralized learning methods in the table cannot observe the important features and their correlations. This is because they are incapable of learning the complicated hidden features using nonlinear transformation and multiple layers of neural network. For the cloud-based deep learning method, even though the RMSE performance is only 2\% higher than the CS-based DFEL due to the same utilization of the deep learning algorithm, this method experiences the remarkable communication overhead and information disclosure. 

To demonstrate the above drawback, we provide the communication overhead and learning speed comparisons for all learning methods in Fig.~\ref{fig:comm_over}. As observed in Fig.~\ref{fig:comm_over}(a), the communication overheads of CS-based DFEL with and without clustering are 96.3\% lower than that of the centralized methods. The reason is that the CSs only require to send the small learned models to each other through the network without sharing any actual datasets. This benefit corresponds to the information disclosure reduction for participating EVs and CSs. Additionally, the use of CS-based DFEL algorithm  can further increase the learning speed performance as shown in Fig.~\ref{fig:comm_over}(b). In this case, the CS-based DFEL without and with clustering can boost the learning speed by 24\% and 42\%, respectively. For the CS-based DFEL with clustering using 2 clusters, it outperforms the cloud-based deep learning and CS-based DFEL methods because it reduces the dataset dimension through learning smaller number of samples in each cluster simultaneously.



\subsection{Economic Model Performance}

\subsubsection{The validity of IR and IC constraints} 

In  Fig.~\ref{fig:util_SGP_IC}, we observe the IR and IC constraints of the SGP for the proposed, i.e., non-collaborative contract, optimization problem. From Fig.~\ref{fig:util_SGP_IC}(a), the SGP always holds non-negative utility for all possible types of the SGP, and thus satisfies the original IR constraints in Eq.~(\ref{eqn:for9b}). As shown in the figure, the utility of the SGP increases monotonically with regard to its own type. As such, the SGP with a higher type will produce a higher utility. This is because the participating CSs can request more energy transfer from the SGP due to the fact that a higher type SGP utilizes a larger energy capacity. Correspondingly, the CSs require to offer higher payments to the SGP which then help raising the SGP's utility. Next, from Fig.~\ref{fig:util_SGP_IC}(b)-(c), we show that the SGP always achieves the highest utility when it applies the appropriate contract determined for its own type. Alternatively, the SGP may degrade its utility performance when choosing unsuitable contracts for its type. In this way, the original IC constraints in Eq.~(\ref{eqn:for9c}) are satisfied. For example, the SGP with type 1, type 5, and type 10 will achieve the highest utility when it uses the appropriate contracts for type 1, type 5, and type 8, respectively. Since we guarantee that the IR and IC constraints of the SGP are achieved, we can then find the feasible contracts for all CSs.


\subsubsection{The social welfare of the network and utilities of CSs}

\begin{figure}[t]
	\centering
	\includegraphics[scale=0.33]{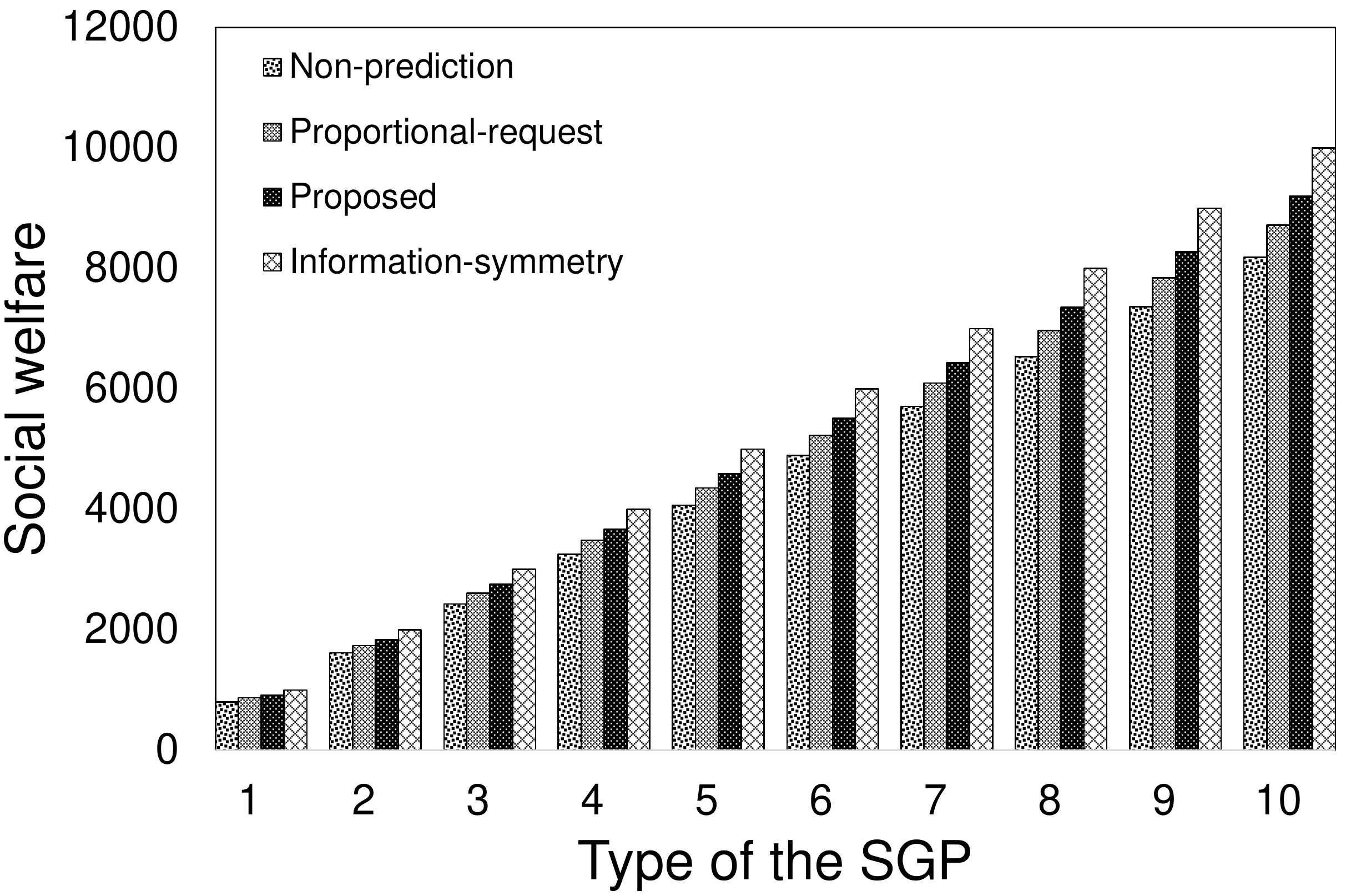}
	\caption{Social welfare between the SGP and all CSs for various methods.}
	\label{fig:social_welfare}
\end{figure}

We then demonstrate the social welfare of the EV network in Fig.~\ref{fig:social_welfare}. In particular, we can enhance the social welfare utilizing the SGP and 58 CSs when the type of the SGP gets higher. Compared with the proportional-request and non-prediction methods, the proposed method can improve the social welfare up to 6\% and 15\%, respectively. Moreover, the social welfare of our non-collaborative contract method is 9\% lower than that of information-symmetry contract method, which works as the upper bound solution. The reason is that, in practice, the CSs do not exactly know the actual energy capacity of the SGP. Hence, the SGP may use lower utilization of its energy capacity to serve the CSs which then incurs lower social welfare of the network.


\begin{figure}[!]
	\begin{center}
		$\begin{array}{cc} 
		\epsfxsize=1.45 in \epsffile{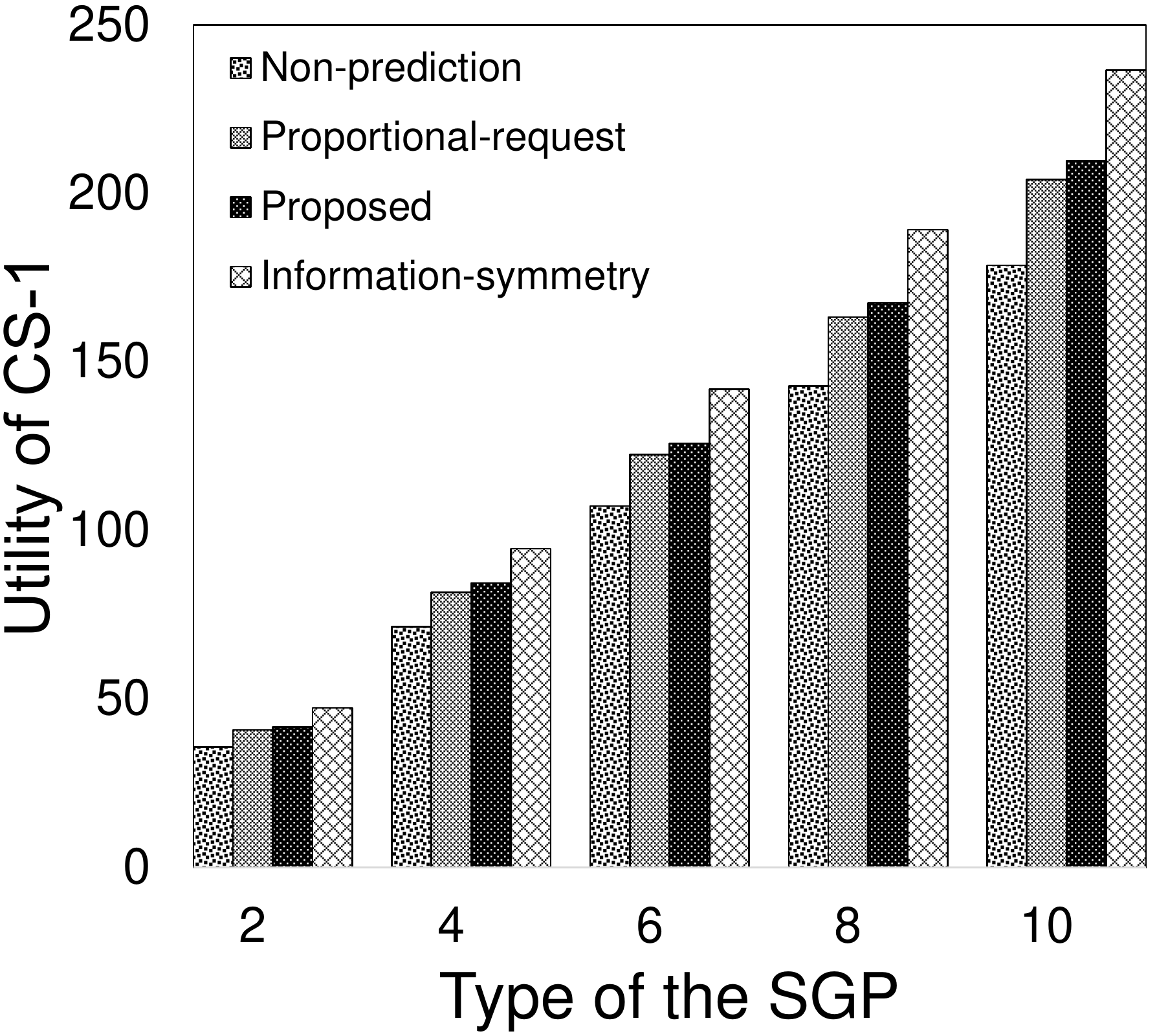} &
		\hspace*{-.2cm}
		\epsfxsize=1.45 in \epsffile{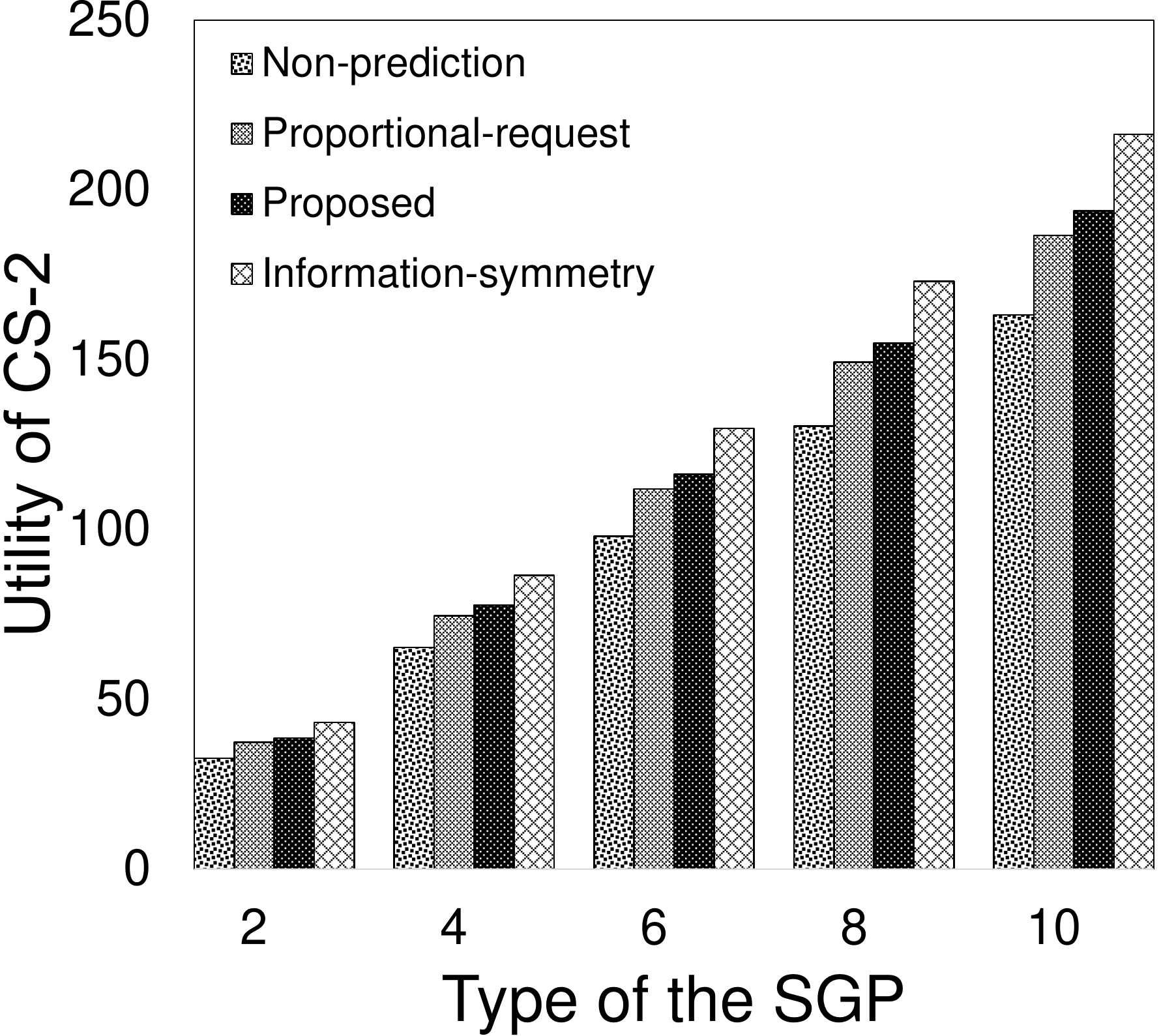} \\ [0.1cm]
		\text{\footnotesize (a) CS-1} & \text{\footnotesize (b) CS-2} \\ [0.2cm]
		\end{array}$
		$\begin{array}{cc} 
		\epsfxsize=1.45 in \epsffile{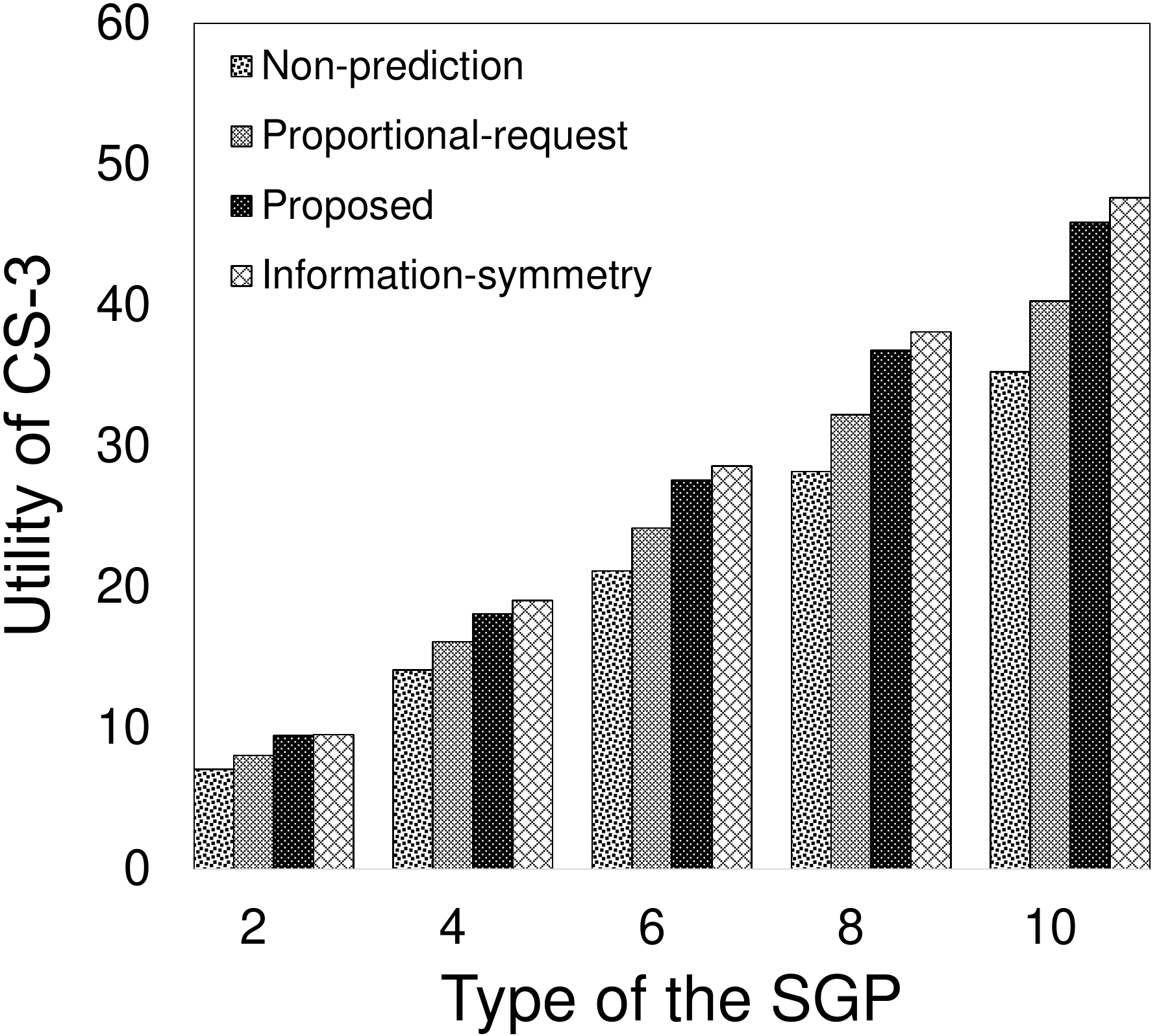} &
		\hspace*{-.2cm}
		\epsfxsize=1.45 in \epsffile{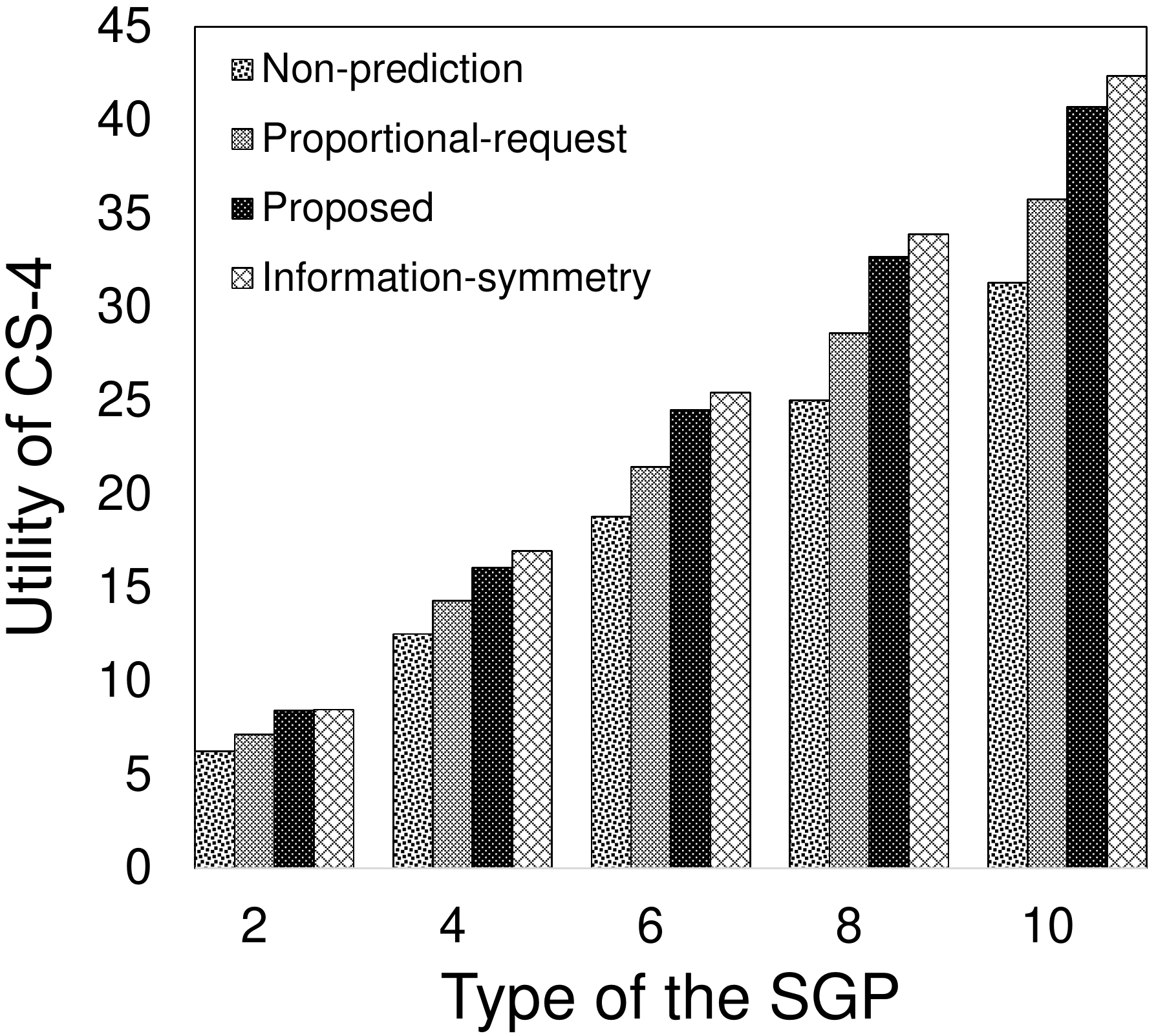} \\ [0.1cm]
		\text{\footnotesize (c) CS-3} & \text{\footnotesize (d) CS-4} \\ [0.2cm]
		\end{array}$
		$\begin{array}{cc} 
		\epsfxsize=1.45 in \epsffile{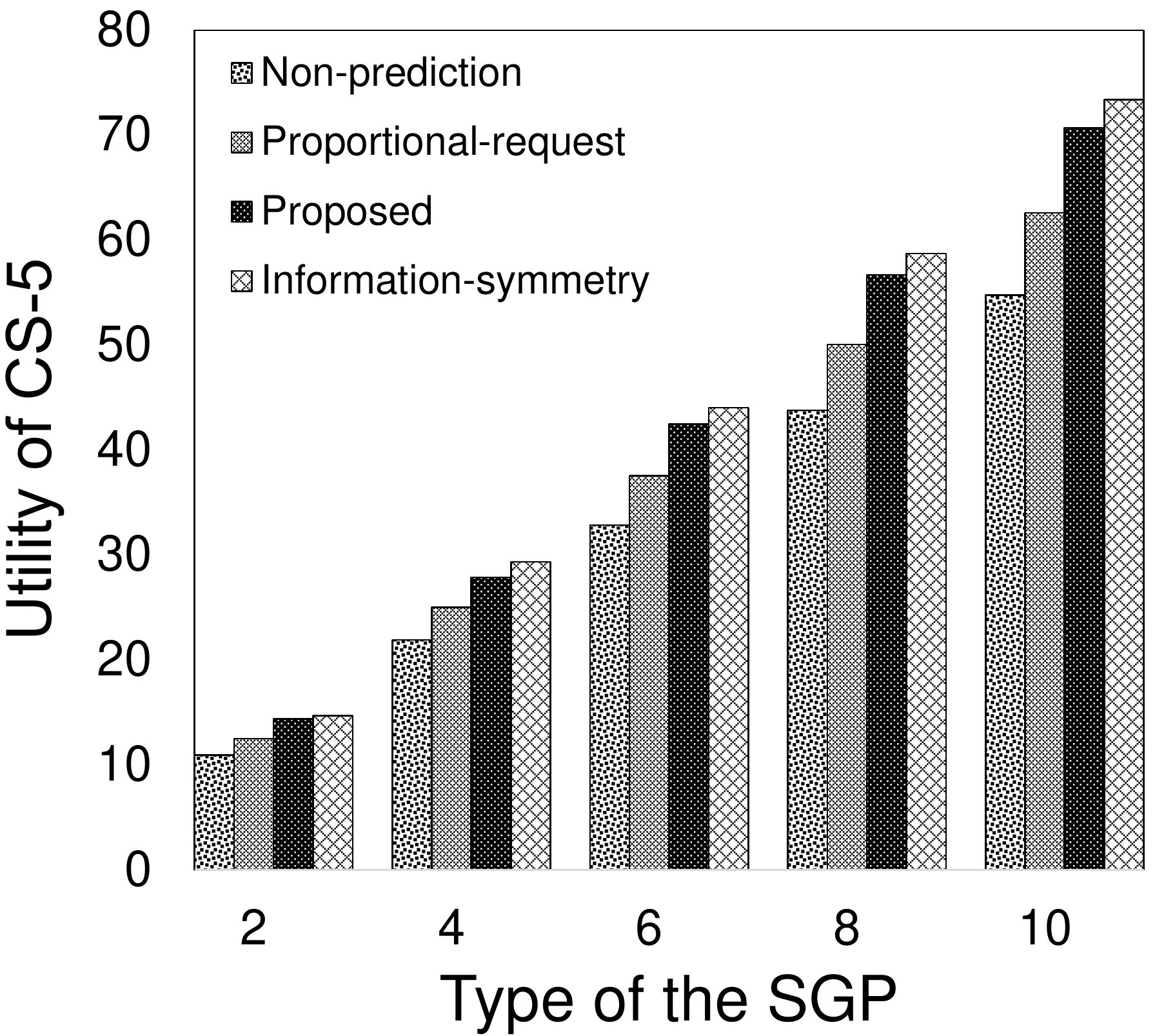} & 
		\hspace*{-.2cm}
		\epsfxsize=1.45 in \epsffile{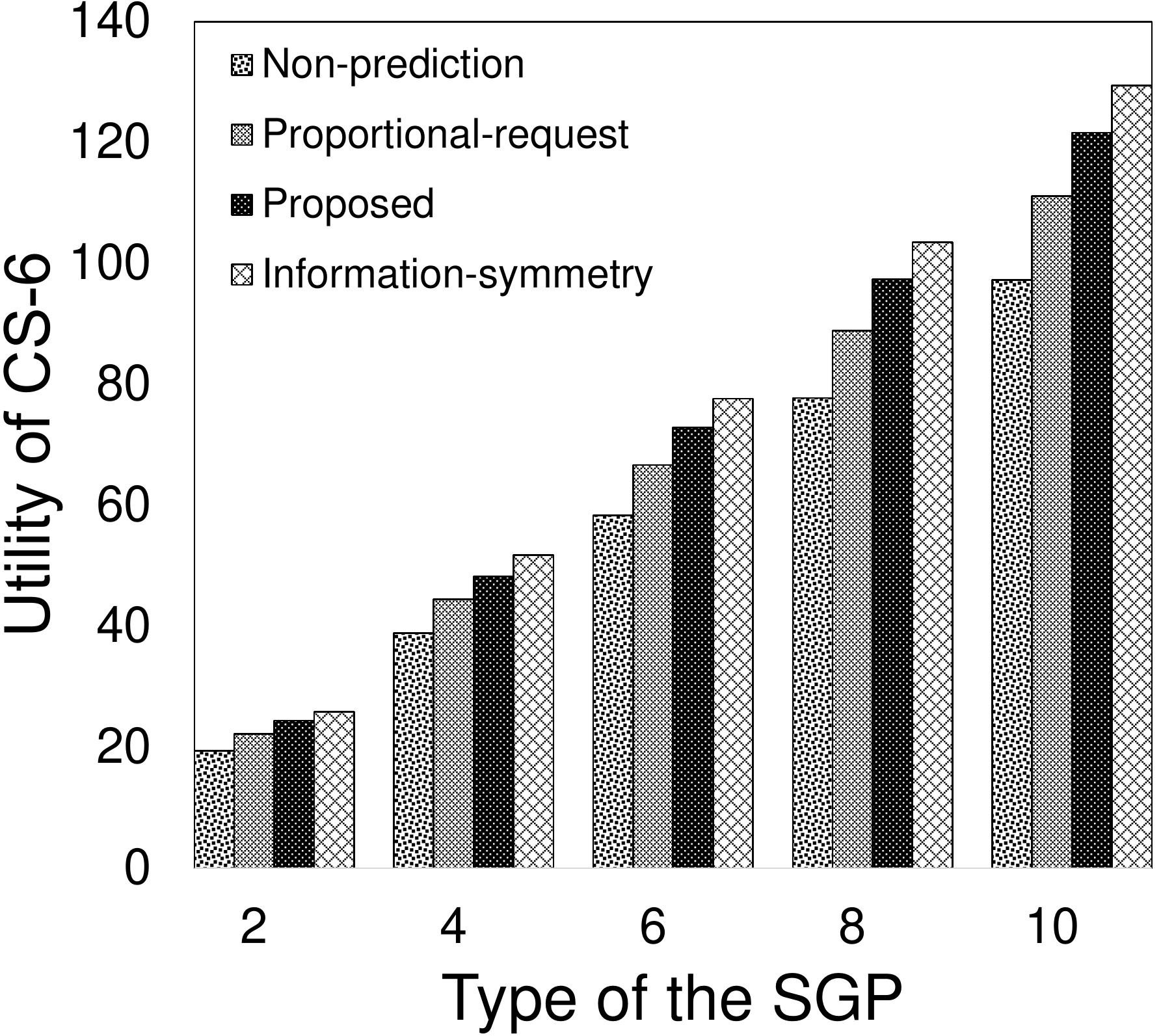} \\ [0.1cm]
		\text{\footnotesize (e) CS-5} & \text{\footnotesize (f) CS-6} \\ [0cm]
		\end{array}$
		\caption{Various utilities of CSs for different methods.}
		\label{fig:util_CS2}
	\end{center}
\end{figure}


To further show the superiority of our proposed method, in Fig.~\ref{fig:util_CS2}, we analyze the utilities of the first 6 CSs as the representative CSs for 10 possible types of the SGP. We observe that the energy demands of CS-1, CS-2, CS-6 (referred to as \emph{high-demand CSs}) are higher than those of CS-3, CS-4, and CS-5 (referred to as \emph{low-demand CSs}). In particular, the proposed method can leverage the utilities of the high-demand and low-demand CSs up to 10\% and 18\%, respectively, compared with those of the proportional-request method. The reason is that, for the proposed method, all CSs can utilize the optimal proportion from the SGP based on the SGP's IR and IC constraints to obtain the optimal contracts. On the other hand, the proportional-request method does not consider the contract policy as well as the SGP's common constraints. Thus, the proportions for all the CSs and the CSs' contracts cannot be optimized to produce the best contracts. We can also observe that the utilities of the high-demand and low-demand CSs are 26\% and 35\% higher than those of the non-prediction method. This is because the non-prediction method experiences fluctuated daily energy prices (with high probability of more expensive price) when the CSs directly request energy from the SGP without predicting the energy demands~\cite{Roman:2011}.

\subsubsection{The expected utilities of CSs vs number of the SGP's types}

\begin{figure}[t]
	\centering
	\includegraphics[scale=0.45]{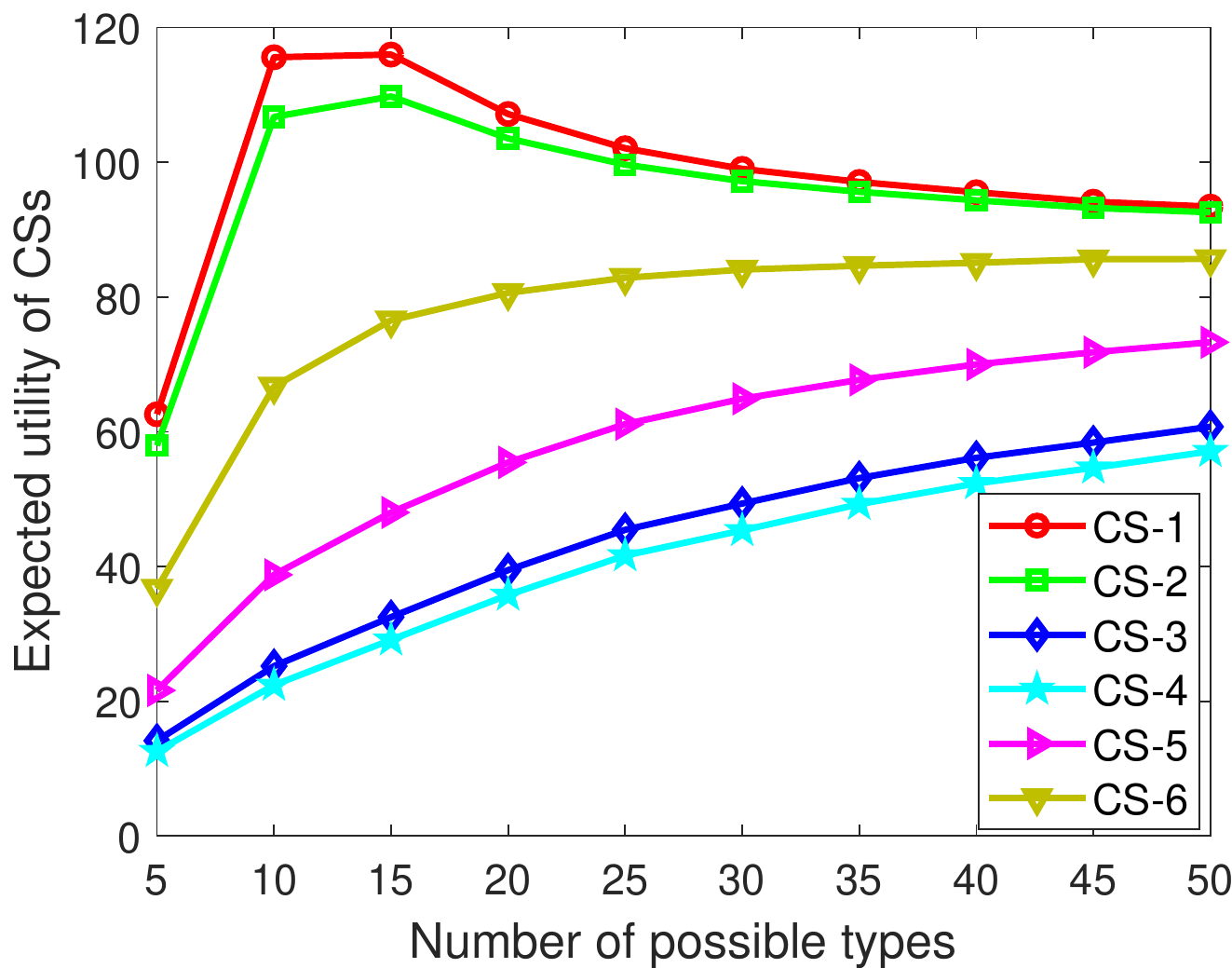}
	\caption{The expected utilities of CSs when the number of SGP's possible types increases.}
	\label{fig:type_increases}
\end{figure}

We then evaluate the proposed method performance as the number of the SGP's possible types increases between 5 and 50 types, with the SGP's true type remains fixed at type 5. This observation can be considered as the expected utility performance when energy capacity of the SGP gets smaller for the same true type of the SGP. 
Specifically, the expected utilities of high-demand CSs fluctuate until the certain number of possible types, i.e., 20 possible types. Then, those expected utilities degrade gradually when we further increase the number of possible types. The reason is that the high-demand CSs cannot find the type of the SGP accurately due to the large range of the energy capacity when smaller number of possible types is taken into account. The same reason also applies for the low-demand CSs when they keep increasing the expected utilities moderately until the case of 15 possible types. However, when the number of possible types becomes bigger, we can observe that all the CSs will achieve the almost-converged expected utility performance at the maximum considered possible types of the SGP. In this way, the accuracy to find the true type of the SGP gets higher as the smaller range of the energy capacity is obtained.

\subsubsection{The total utilities of CSs vs the number of energy transfer price units}

\begin{figure}[t]
	\centering
	\includegraphics[scale=0.33]{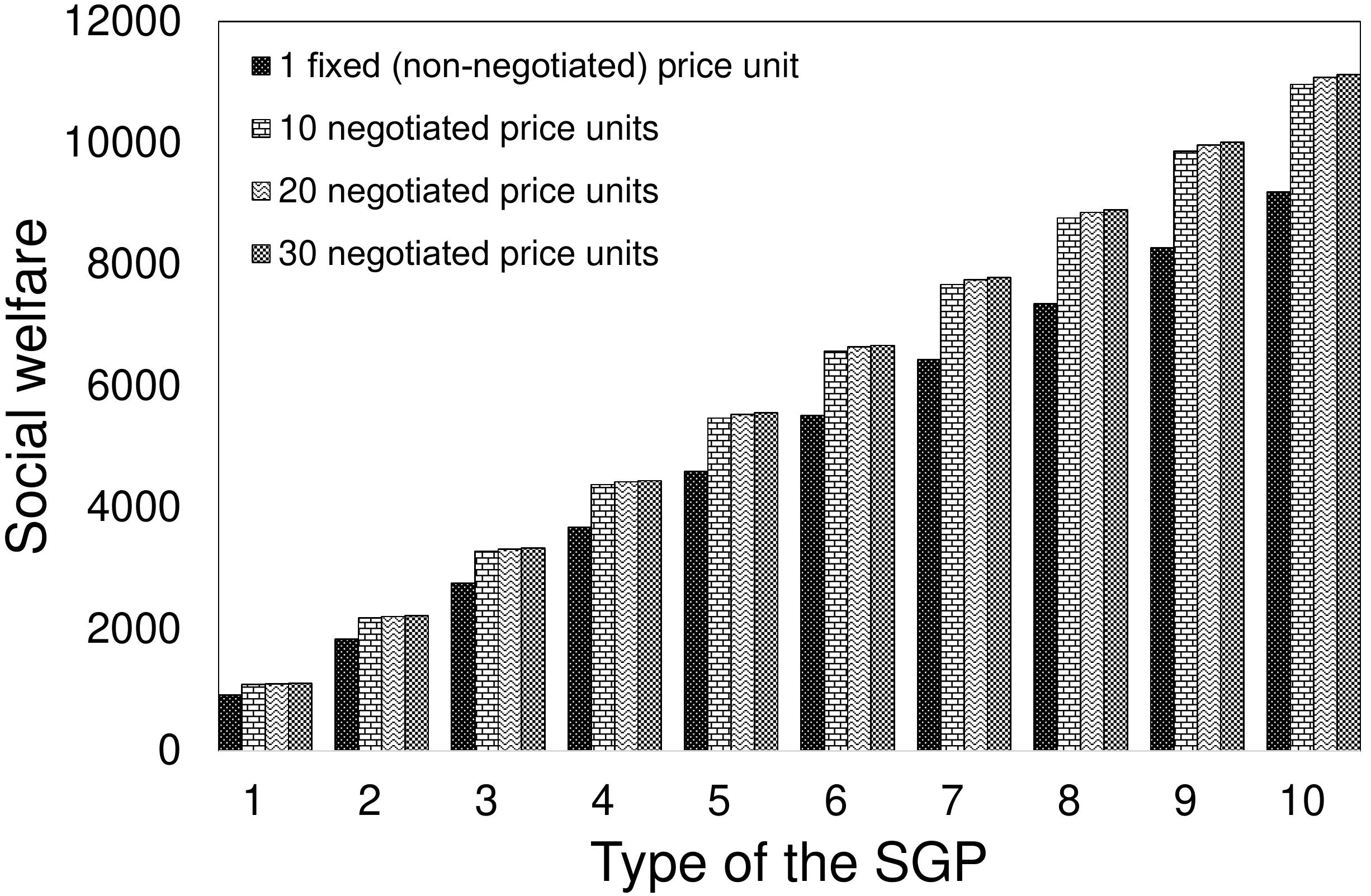}
	\caption{Social welfare of the proposed method for various number of negotiated energy transfer price units.}
	\label{fig:social_welfare2}
\end{figure}

\begin{figure}[!]
	\begin{center}
		$\begin{array}{cc} 
		\epsfxsize=1.65 in \epsffile{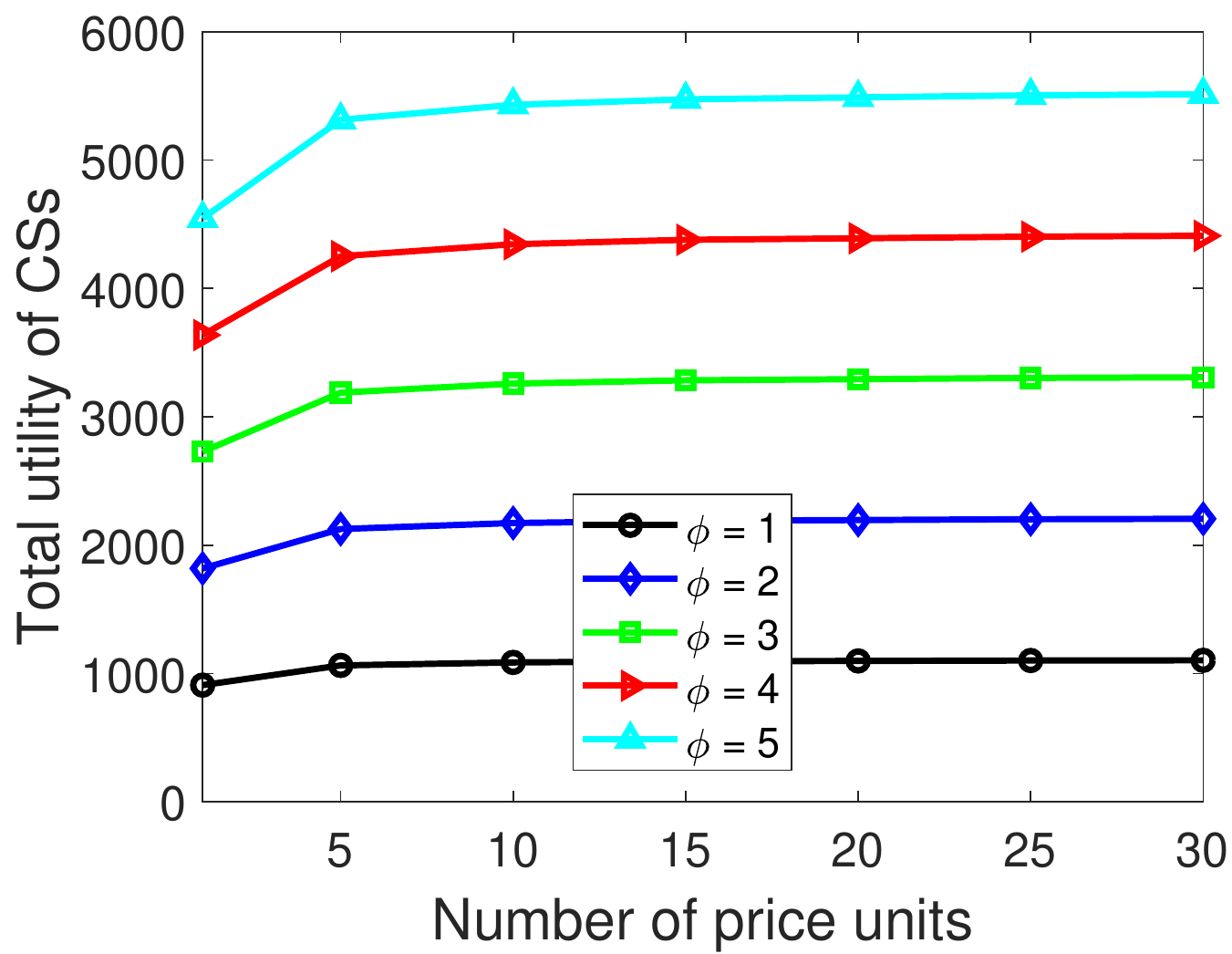} & 
		\hspace*{-0.2cm}
		\epsfxsize=1.65 in \epsffile{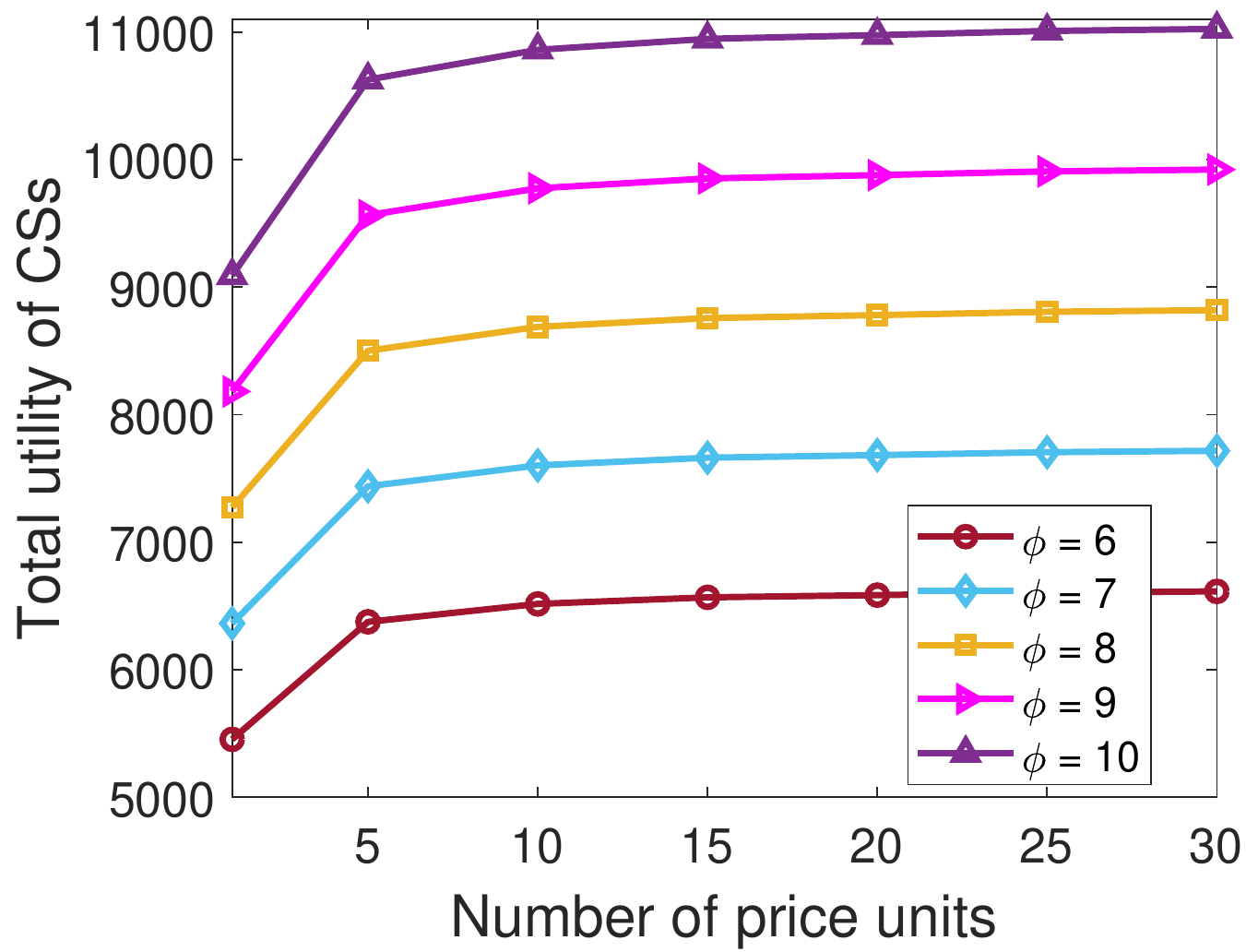} \\ [0.1cm]
		\text{\footnotesize (a) $\phi = 1$ to $\phi = 5$} & \text{\footnotesize (b) $\phi = 6$ to $\phi = 10$} \\ [0.2cm]
		\end{array}$
		\caption{The total utilities of all CSs when the number of negotiated energy transfer price units increases.}
		\label{fig:price_utilities}
	\end{center}
\end{figure}

To prove that the proposed economic model is always more flexible than other non-contract-based methods, i.e., proportional-request and non-prediction methods, we evaluate the social welfare and total CSs' utilities of the proposed method in Fig.~\ref{fig:social_welfare2} and~\ref{fig:price_utilities}, respectively. As such, the SGP and CSs can negotiate together to change the energy transfer price units. 
Particularly, when some CSs request high energy transfer from the SGP, the SGP can slightly reduce the price unit. Meanwhile, the SGP can increase the price unit when only small amount of energy transfer is requested by the other CSs. In this way, the SGP can attract more low-demand CSs to compete in the contract process with higher energy transfer proportions, thereby improving the social welfare of the network. 

To this end, we show the social welfare improvement in Fig.~\ref{fig:social_welfare2} for different number of negotiated energy transfer price units between 190MU and 200MU. As observed in the figure, the variation of the price units (between 10 and 30 price units) for energy transfer negotiation can further boost the social welfare of the network up to 21\% compared with the scenario when only one fixed price unit is applied for all energy transfer requests. This trend aligns with that of the CSs' total utilities as shown in Fig.~\ref{fig:price_utilities}. When the number of energy transfer price units is only one, the total utilities for all possible types suffer from the worst total utility performance. However, when we further vary the number of negotiated energy transfer price units, we can improve the total utilities of the CSs for all possible types until reaching the convergence. These results clearly demonstrate that using the proposed economic model can always outperform other conventional economic models (where the CSs only can modify their utilities based on the given price unit from the SGP without any negotiation) because of the flexibility in the contract negotiation process.

\section{Conclusion}
\label{sec:Conc}
In this paper, we have proposed the novel energy-efficient framework leveraging the effective contract theoretic-based economic model to maximize the profits of CSs and improve the social welfare in the EV network. In particular, we have introduced the CS-based decentralized federated energy learning (DFEL) framework which allows the CSs learning the local dataset to predict energy demands accurately and reduce the communication overhead significantly. Moreover, we have developed the CS clustering-based DFEL approach for the CSs to further boost the energy demand prediction accuracy and reduce the learning time due to the smaller dataset dimension. Based on this prediction, we have designed the MPOA contract-based economic model. In particular, we have formulated the contract-based problem as the non-collaborative energy contract optimization which satisfies the common constraints, i.e., the individual rationality and incentive compatibility, from the SGP. This aims to maximize the expected utility of the CSs and improve the social welfare in the energy transfer process. Then, we have developed the iterative energy
contract algorithm to achieve the equilibrium contract solution for all the CSs. Through simulation results, we have shown that our proposed method outperforms other centralized learning algorithms in terms of the prediction accuracy, communication overhead, and learning speed. Moreover, through the proposed method, we can significantly improve the utilities of the CSs and the social welfare of the network compared with other non-contract-based economic models.



\ifCLASSOPTIONcaptionsoff
\newpage
\fi

\vspace{0.7cm}

\clearpage
\appendices

\section{Proof of Lemma 1}
\label{appx:lemma0a}

From the constraints (\ref{eqn:for9a})-(\ref{eqn:for9c}), we can compute the total number of constraints according to the considered number of possible types $\phi_{tot}$. Specifically, there are $\phi_{tot}$ number of energy capacity constraints (\ref{eqn:for9a}) and $\phi_{tot}$ number of the IR constraints (\ref{eqn:for9b}). In addition, there exist $\phi_{tot}(\phi_{tot}-1)$ number of IC constraints (\ref{eqn:for9c}). Hence, the total number of constraints is $\Big(\phi_{tot} + \phi_{tot} + \phi_{tot}(\phi_{tot}-1)\Big) = \Big(\phi_{tot} + \phi_{tot}^2\Big)$. Since $\max\Big(\phi_{tot},\phi_{tot}^2\Big) = \phi_{tot}^2$, then the computational complexity of problem $(\mathbf{P}_3)$ is $O(\phi_{tot}^2)$.

\section{Proof of Lemma 2}
\label{appx:lemma1}

Using the IC constraint from the Definition~\ref{DefV2}, we first prove that $\boldsymbol{\rho}(\phi) > \boldsymbol{\rho}(\phi^*)$ if and only if $\phi > \phi^*$. Based on Eq.~(\ref{eqn:for7b}), we now have
\begin{equation}
\label{eqn:for7c}
\begin{aligned}
\phi G(\boldsymbol{\hat \pi},\boldsymbol{\rho}(\phi)) - C(\boldsymbol{\hat \pi},\boldsymbol{\xi}(\phi)) &\geq \phi G(\boldsymbol{\hat \pi},\boldsymbol{\rho}(\phi^*)) - C(\boldsymbol{\hat \pi},\boldsymbol{\xi}(\phi^*)),
\end{aligned}
\end{equation}
and if the SGP has type $\phi^*$, then we obtain
\begin{equation}
\label{eqn:for7d}
\begin{aligned}
\phi^* G(\boldsymbol{\hat \pi},\boldsymbol{\rho}(\phi^*)) - C(\boldsymbol{\hat \pi},\boldsymbol{\xi}(\phi^*)) \geq \\ \phi^* G(\boldsymbol{\hat \pi},\boldsymbol{\rho}(\phi)) - C(\boldsymbol{\hat \pi},\boldsymbol{\xi}(\phi)),
\end{aligned}
\end{equation}
where $\phi \neq \phi^*$ and $\phi, \phi^* \in \Phi$. By combining Eq.~(\ref{eqn:for7c}) and Eq.~(\ref{eqn:for7d}), we have
\begin{equation}
\label{eqn:for7e}
\begin{aligned}
\phi G(\boldsymbol{\hat \pi},\boldsymbol{\rho}(\phi)) + \phi^* G(\boldsymbol{\hat \pi},\boldsymbol{\rho}(\phi^*)) \geq \\ \phi G(\boldsymbol{\hat \pi},\boldsymbol{\rho}(\phi^*)) + \phi^* G(\boldsymbol{\hat \pi},\boldsymbol{\rho}(\phi)), \\
\phi G(\boldsymbol{\hat \pi},\boldsymbol{\rho}(\phi)) - \phi^* G(\boldsymbol{\hat \pi},\boldsymbol{\rho}(\phi)) \geq \\ \phi G(\boldsymbol{\hat \pi},\boldsymbol{\rho}(\phi^*)) - \phi^* G(\boldsymbol{\hat \pi},\boldsymbol{\rho}(\phi^*)), \\
G(\boldsymbol{\hat \pi},\boldsymbol{\rho}(\phi))\big(\phi - \phi^*\big) \geq G(\boldsymbol{\hat \pi},\boldsymbol{\rho}(\phi^*))\big(\phi - \phi^*\big).
\end{aligned}
\end{equation}
Then, we divide both sides using $(\phi - \phi^*)$, and thus $G(\boldsymbol{\hat \pi},\boldsymbol{\rho}(\phi)) > G(\boldsymbol{\hat \pi},\boldsymbol{\rho}(\phi^*))$ considering that $\phi - \phi^* > 0$ since $\phi > \phi^*$. Based on Eq.~(\ref{eqn:for4}), the gain function is strictly increasing for all $\boldsymbol{\rho}(\phi)$. Consequently, as $G(\boldsymbol{\hat \pi},\boldsymbol{\rho}(\phi)) > G(\boldsymbol{\hat \pi},\boldsymbol{\rho}(\phi^*))$ satisfies, we have $\boldsymbol{\rho}(\phi) > \boldsymbol{\rho}(\phi^*)$.

Next, we prove that $\phi > \phi^*$ if and only if $\boldsymbol{\rho}(\phi) > \boldsymbol{\rho}(\phi^*)$. Using Eqs.~(\ref{eqn:for7c})-(\ref{eqn:for7e}), we get
\begin{equation}
\label{eqn:for7f}
\begin{aligned}
\phi \big(G(\boldsymbol{\hat \pi},\boldsymbol{\rho}(\phi)) - G(\boldsymbol{\hat \pi},\boldsymbol{\rho}(\phi^*))\big) \geq \\ \phi^* \big(G(\boldsymbol{\hat \pi},\boldsymbol{\rho}(\phi)) - G(\boldsymbol{\hat \pi},\boldsymbol{\rho}(\phi^*)).
\end{aligned}
\end{equation}
By dividing both sides with $\big(G(\boldsymbol{\hat \pi},\boldsymbol{\rho}(\phi)) - G(\boldsymbol{\hat \pi},\boldsymbol{\rho}(\phi^*))\big)$, we obtain $\phi > \phi^*$ accounting for $\big(G(\boldsymbol{\hat \pi},\boldsymbol{\rho}(\phi)) - G(\boldsymbol{\hat \pi},\boldsymbol{\rho}(\phi^*))\big) > 0$ since $G(\boldsymbol{\hat \pi},\boldsymbol{\rho}(\phi)) > G(\boldsymbol{\hat \pi},\boldsymbol{\rho}(\phi^*))$. Based on Eq.~(\ref{eqn:for4}), the gain function is strictly increasing for all $\boldsymbol{\rho}(\phi)$. As a result, we prove that if $\boldsymbol{\rho}(\phi) > \boldsymbol{\rho}(\phi^*)$, then $\phi > \phi^*$. The same process can be used to prove that if $\phi = \phi^*$, then $\boldsymbol{\rho}(\phi) = \boldsymbol{\rho}(\phi^*)$. 

\section{Proof of Lemma 3}
\label{appx:lemma2}

We can first prove the conditions in Eq.~(\ref{eqn:for140}) based on the proof from Lemma~\ref{lemma1} along with the following description. From Eq.~(\ref{eqn:for4}), we can compute the first-order form of the gain function by
\begin{equation}
\label{eqn:for1400}
\begin{aligned}
\frac{dG(\boldsymbol{\hat \pi},\boldsymbol{\rho}(\phi))}{d\phi} =\frac{\overset{I}{\underset{i=1}{\sum}} {\hat \pi}_i\frac{d\rho_i(\phi)}{d\phi} }{1 + \overset{I}{\underset{i=1}{\sum}} {\hat \pi}_i\rho_i(\phi)} \geq 0.
\end{aligned}
\end{equation}
Since $0 \leq {\hat \pi}_i \leq 1, \forall i \in \mathcal{I}$, then $\frac{\overset{I}{\underset{i=1}{\sum}} {\hat \pi}_i}{1 + \overset{I}{\underset{i=1}{\sum}} {\hat \pi}_i\rho_i(\phi)} \geq 0$. As a result, $\frac{d\rho_i(\phi)}{d\phi} \geq 0$ satisfies the first-order condition in Eq.~(\ref{eqn:for1400}).

Next, for the conditions in Eq.~(\ref{eqn:for14}), we can use the contradiction to make the IC constraint cannot be satisfied. In this case, there is at least one $\phi'$ fails to comply with the IC constraint such that
\begin{equation}
\label{eqn:for14a}
\begin{aligned}
0 \leq \phi G(\boldsymbol{\hat \pi},\boldsymbol{\rho}(\phi)) - C(\boldsymbol{\hat \pi},\boldsymbol{\xi}(\phi)) < \\ \phi G(\boldsymbol{\hat \pi},\boldsymbol{\rho}(\phi')) - C(\boldsymbol{\hat \pi},\boldsymbol{\xi}(\phi')),
\end{aligned}
\end{equation}
where $\phi < \phi'$. Alternatively, using the integration from $\phi$ to $\phi'$, we obtain
\begin{equation}
\label{eqn:for14b2}
\begin{aligned}
\int_{\phi}^{\phi'}\bigg(\phi \frac{dG(\boldsymbol{\hat \pi},\boldsymbol{\rho}(l))}{dl} - \frac{dC(\boldsymbol{\hat \pi},\boldsymbol{\xi}(l))}{dl}\bigg) dl > 0.
\end{aligned}
\end{equation}
From Eq.~(\ref{eqn:for14}), we have $\int_{\phi}^{\phi'}\bigg(l  \frac{dG(\boldsymbol{\hat \pi},\boldsymbol{\rho}(l))}{dl} - \frac{dC(\boldsymbol{\hat \pi},\boldsymbol{\xi}(l))}{dl}\bigg) dl = 0$. If $\phi < l < \phi'$, then 
\begin{equation}
\label{eqn:for14d}
\begin{aligned}
\phi \frac{dG(\boldsymbol{\hat \pi},\boldsymbol{\rho}(l))}{dl} \leq l  \frac{dG(\boldsymbol{\hat \pi},\boldsymbol{\rho}(l))}{dl}.
\end{aligned}
\end{equation}
Hence, 
\begin{equation}
\label{eqn:for14e}
\begin{aligned}
\int_{\phi}^{\phi'}\bigg(\phi \frac{dG(\boldsymbol{\hat \pi},\boldsymbol{\rho}(l))}{dl} - \frac{dC(\boldsymbol{\hat \pi},\boldsymbol{\xi}(l))}{dl}\bigg) dl < 0,
\end{aligned}
\end{equation}
and we can see the contradiction. Likewise, if $\phi > \phi'$, we can also obtain the contradiction. In this case,
\begin{equation}
\label{eqn:for14f}
\begin{aligned}
0 \leq \phi' G(\boldsymbol{\hat \pi},\boldsymbol{\rho}(\phi')) - C(\boldsymbol{\hat \pi},\boldsymbol{\xi}(\phi')) < \\ \phi' G(\boldsymbol{\hat \pi},\boldsymbol{\rho}(\phi)) - C(\boldsymbol{\hat \pi},\boldsymbol{\xi}(\phi)),
\end{aligned}
\end{equation}
and
\begin{equation}
\label{eqn:for14b}
\begin{aligned}
\int_{\phi'}^{\phi}\bigg(\phi' \frac{dG(\boldsymbol{\hat \pi},\boldsymbol{\rho}(l))}{dl} - \frac{dC(\boldsymbol{\hat \pi},\boldsymbol{\xi}(l))}{dl}\bigg) dl > 0.
\end{aligned}
\end{equation}
Similarly, we also have $\int_{\phi'}^{\phi}\bigg(l \frac{dG(\boldsymbol{\hat \pi},\boldsymbol{\rho}(l))}{dl} - \frac{dC(\boldsymbol{\hat \pi},\boldsymbol{\xi}(l))}{dl}\bigg) dl = 0$. If $\phi' < l < \phi$, then 
\begin{equation}
\label{eqn:for14g}
\begin{aligned}
\phi' \frac{dG(\boldsymbol{\hat \pi},\boldsymbol{\rho}(l))}{dl} \leq l  \frac{dG(\boldsymbol{\hat \pi},\boldsymbol{\rho}(l))}{dl}.
\end{aligned}
\end{equation}
and
\begin{equation}
\label{eqn:for14h}
\begin{aligned}
\int_{\phi'}^{\phi}\bigg(\phi' \frac{dG(\boldsymbol{\hat \pi},\boldsymbol{\rho}(l))}{dl} - \frac{dC(\boldsymbol{\hat \pi},\boldsymbol{\xi}(l))}{dl}\bigg) dl < 0.
\end{aligned}
\end{equation}
Hence, we also obtain the contradiction.

\section{Proof of Lemma 4}
\label{appx:lemma0b}

From the constraints (\ref{eqn:for9a}), (\ref{eqn:for9revb}), (\ref{eqn:for9revc2}), and (\ref{eqn:for9revd2}), there exist $\phi_{tot}$ number of energy capacity constraints (\ref{eqn:for9a}), one constraint (\ref{eqn:for9revb}), $\phi_{tot}$ number of the constraints (\ref{eqn:for9revc}), and $\phi_{tot}$ number of the constraints (\ref{eqn:for9revd}). Thus, the total number of constraints is $\Big(\phi_{tot} + 1 + \phi_{tot} + \phi_{tot})\Big) = \Big(1 + 3\phi_{tot}\Big)$. As $\max\Big(\phi_{tot}, 1,\phi_{tot},\phi_{tot}\Big) = \phi_{tot}$, then the computational complexity of problem $(\mathbf{P}_5)$ is $O(\phi_{tot})$.

\section{Proof of Theorem 1}
\label{appx:theorem1}

Suppose $\mathbb{C}$ contains $\Big(\boldsymbol{\rho}_{-i}^{(\theta+1)}(\phi), \boldsymbol{\xi}^{(\theta+1)}_{-i}(\phi)\Big)$ and $\Big(\boldsymbol{\rho}^{(\theta)}_{-i}(\phi), \boldsymbol{\xi}^{(\theta)}_{-i}(\phi)\Big)$ which are ordered by
\begin{equation}
\label{eqn:for16a}
\begin{aligned}
\Big(\boldsymbol{\rho}^{(\theta+1)}_{-i}(\phi), \boldsymbol{\xi}^{(\theta+1)}_{-i}(\phi)\Big) > \Big(\boldsymbol{\rho}^{(\theta)}_{-i}(\phi), \boldsymbol{\xi}^{(\theta)}_{-i}(\phi)\Big).
\end{aligned}
\end{equation}
We first show that the choice ${\hat \Gamma}_i$ in $\Gamma_i$ at CS-$i$, where ${\hat \Gamma}^{(\theta+1)}_i\Big(\boldsymbol{\rho}^{(\theta)}_{-i}(\phi), \boldsymbol{\xi}^{(\theta)}_{-i}(\phi)\Big) = \max \Gamma^{(\theta+1)}_i\Big(\boldsymbol{\rho}^{(\theta)}_{-i}(\phi), \boldsymbol{\xi}^{(\theta)}_{-i}(\phi)\Big)$, is increasing, i.e., 
\begin{equation}
\label{eqn:for16b}
\begin{aligned}
{\hat \Gamma}^{(\theta+2)}_i\Big(\boldsymbol{\rho}^{(\theta+1)}_{-i}(\phi), \boldsymbol{\xi}^{(\theta+1)}_{-i}(\phi)\Big) \geq {\hat \Gamma}^{(\theta+1)}_i\Big(\boldsymbol{\rho}^{(\theta)}_{-i}(\phi), \boldsymbol{\xi}^{(\theta)}_{-i}(\phi)\Big).
\end{aligned}
\end{equation}
Let $\Big(\boldsymbol{\rho}^{(\theta+1)}_{-i}(\phi), \boldsymbol{\xi}^{(\theta+1)}_{-i}(\phi)\Big) > \Big(\boldsymbol{\rho}^{(\theta)}_{-i}(\phi), \boldsymbol{\xi}^{(\theta)}_{-i}(\phi)\Big)$ and ${\hat \Gamma}^{(\theta+2)}_i\Big(\boldsymbol{\rho}^{(\theta+1)}_{-i}(\phi), \boldsymbol{\xi}^{(\theta+1)}_{-i}(\phi)\Big) < {\hat \Gamma}^{(\theta+1)}_i\Big(\boldsymbol{\rho}^{(\theta)}_{-i}(\phi), \boldsymbol{\xi}^{(\theta)}_{-i}(\phi)\Big)$ to show a contradiction. Then, we have
\begin{equation}
\label{eqn:for16c}
\begin{aligned}
{U}_i\bigg({\hat \Gamma}^{(\theta+1)}_i\Big(\boldsymbol{\rho}^{(\theta)}_{-i}(\phi), \boldsymbol{\xi}^{(\theta)}_{-i}(\phi)\Big),\boldsymbol{\rho}^{(\theta+2)}_{-i}(\phi), \boldsymbol{\xi}^{(\theta+2)}_{-i}(\phi)\bigg) + \\ {U}_i\bigg({\hat \Gamma}^{(\theta+2)}_i\Big(\boldsymbol{\rho}^{(\theta+1)}_{-i}(\phi), \boldsymbol{\xi}^{(\theta+1)}_{-i}(\phi)\Big),\boldsymbol{\rho}^{(\theta+1)}_{-i}(\phi), \boldsymbol{\xi}^{(\theta+1)}_{-i}(\phi)\bigg) \geq \\ {U}_i\bigg({\hat \Gamma}^{(\theta+2)}_i\Big(\boldsymbol{\rho}^{(\theta+1)}_{-i}(\phi), \boldsymbol{\xi}^{(\theta+1)}_{-i}(\phi)\Big),\boldsymbol{\rho}^{(\theta+2)}_{-i}(\phi), \boldsymbol{\xi}^{(\theta+2)}_{-i}(\phi)\bigg) + \\ {U}_i\bigg({\hat \Gamma}^{(\theta+1)}_i\Big(\boldsymbol{\rho}^{(\theta)}_{-i}(\phi), \boldsymbol{\xi}^{(\theta)}_{-i}(\phi)\Big),\boldsymbol{\rho}^{(\theta+1)}_{-i}(\phi), \boldsymbol{\xi}^{(\theta+1)}_{-i}(\phi)\bigg).
\end{aligned}
\end{equation}
From the definition of $\Gamma^{(\theta+1)}_i$, ${\hat \Gamma}^{(\theta+1)}_i\Big(\boldsymbol{\rho}^{(\theta)}_{-i}(\phi), \boldsymbol{\xi}^{(\theta)}_{-i}(\phi)\Big) \in \Gamma^{(\theta+1)}_i\Big(\boldsymbol{\rho}^{(\theta)}_{-i}(\phi), \boldsymbol{\xi}^{(\theta)}_{-i}(\phi)\Big)$ specifies that
\begin{equation}
\label{eqn:for16d}
\begin{aligned}
{U}_i\bigg({\hat \Gamma}^{(\theta+1)}_i\Big(\boldsymbol{\rho}^{(\theta)}_{-i}(\phi), \boldsymbol{\xi}^{(\theta)}_{-i}(\phi)\Big),\boldsymbol{\rho}^{(\theta+1)}_{-i}(\phi), \boldsymbol{\xi}^{(\theta+1)}_{-i}(\phi)\bigg) \geq \\ {U}_i\bigg({\hat \Gamma}^{(\theta+2)}_i\Big(\boldsymbol{\rho}^{(\theta+1)}_{-i}(\phi), \boldsymbol{\xi}^{(\theta+1)}_{-i}(\phi)\Big),\boldsymbol{\rho}^{(\theta+1)}_{-i}(\phi), \boldsymbol{\xi}^{(\theta+1)}_{-i}(\phi)\bigg).
\end{aligned}
\end{equation}
Then, based on Eq.~(\ref{eqn:for16c}) and Eq.~(\ref{eqn:for16d}), we can obtain
\begin{equation}
\label{eqn:for16e}
\begin{aligned}
{U}_i\bigg({\hat \Gamma}^{(\theta+1)}_i\Big(\boldsymbol{\rho}^{(\theta)}_{-i}(\phi), \boldsymbol{\xi}^{(\theta)}_{-i}(\phi)\Big),\boldsymbol{\rho}^{(\theta+2)}_{-i}(\phi), \boldsymbol{\xi}^{(\theta+2)}_{-i}(\phi)\bigg) \geq \\ {U}_i\bigg({\hat \Gamma}^{(\theta+2)}_i\Big(\boldsymbol{\rho}^{(\theta+1)}_{-i}(\phi), \boldsymbol{\xi}^{(\theta+1)}_{-i}(\phi)\Big),\boldsymbol{\rho}^{(\theta+2)}_{-i}(\phi), \boldsymbol{\xi}^{(\theta+2)}_{-i}(\phi)\bigg),
\end{aligned}
\end{equation}
and thus ${\hat \Gamma}^{(\theta+1)}_i\Big(\boldsymbol{\rho}^{(\theta)}_{-i}(\phi), \boldsymbol{\xi}^{(\theta)}_{-i}(\phi)\Big) \in \Gamma^{(\theta+2)}_i\Big(\boldsymbol{\rho}^{(\theta+1)}_{-i}(\phi), \boldsymbol{\xi}^{(\theta+1)}_{-i}(\phi)\Big)$. Alternatively, it indicates that ${\hat \Gamma}^{(\theta+1)}_i\Big(\boldsymbol{\rho}^{(\theta)}_{-i}(\phi), \boldsymbol{\xi}^{(\theta)}_{-i}(\phi)\Big) > {\hat \Gamma}^{(\theta+2)}_i\Big(\boldsymbol{\rho}^{(\theta+1)}_{-i}(\phi), \boldsymbol{\xi}^{(\theta+1)}_{-i}(\phi)\Big)$. From the definition of ${\hat \Gamma}^{(\theta+1)}_i$, it follows that ${\hat \Gamma}^{(\theta+2)}_i\Big(\boldsymbol{\rho}^{(\theta+1)}_{-i}(\phi), \boldsymbol{\xi}^{(\theta+1)}_{-i}(\phi)\Big) \geq {\hat \Gamma}^{(\theta+1)}_i\Big(\boldsymbol{\rho}^{(\theta)}_{-i}(\phi), \boldsymbol{\xi}^{(\theta)}_{-i}(\phi)\Big)$ which produces a contradiction. Hence, ${\hat \Gamma}_i$ is an increasing function. 

Based on the aforementioned increasing function, the contract of CS-$i$ at iteration $\theta + 1$ can be updated to $\Big(\rho^{(\theta+1)}_i(\phi),\xi^{(\theta+1)}_i(\phi)\Big) \in \Gamma^{(\theta+1)}_i$ if the condition in Eq.~(\ref{eqn:for11c}) is satisfied. In this case, the use of $\kappa$ can be seen as the optimality tolerance to terminate the iterative algorithm. Consequently, when all the CSs do not hold the condition in Eq.~(\ref{eqn:for11c}) at iteration $\theta= \vartheta$, then we will obtain
\begin{equation}
\label{eqn:for16j}
\begin{aligned}
\Big(\rho_i^{(\vartheta+1)}(\phi),\xi_i^{(\vartheta+1)}(\phi)\Big) = \Big(\rho_i^{(\vartheta)}(\phi),\xi_i^{(\vartheta)}(\phi)\Big), \\ \forall i \in \mathcal{I}, \forall \phi \in \Phi.
\end{aligned}
\end{equation}
Consequently, for the rest of $\theta$ values starting from $\vartheta$, the algorithm obtains the same ${U}_{i}\Big(\rho^{(\theta+1)}_i(\phi),\xi^{(\theta+1)}_i(\phi),\boldsymbol{\rho}^{(\theta)}_{-i}(\phi), \boldsymbol{\xi}^{(\theta)}_{-i}(\phi)\Big), \forall i \in \mathcal{I}$. It implies that the algorithm converges under the optimality tolerance $\kappa$.


\section{Proof of Theorem 2}
\label{appx:theorem2}

We adopt this proof from~\cite{Bernheim:1986} and \cite{Fraysse:1993}. In particular, we first show that an equilibrium exists through obtaining a fixed point of $\Gamma$. Consider $\mathbb{C}^*$ which contains $\Big(\boldsymbol{\rho}(\phi), \boldsymbol{\xi}(\phi)\Big) \in \mathbb{C}$. 
This $\mathbb{C}^*$ is a non-empty contract space since ${\hat \Gamma}\Big(\min \mathbb{C}\Big) \geq \min \mathbb{C}$, ${\hat \Gamma} \in \Gamma$, and thus
\begin{equation}
\label{eqn:for16f}
\begin{aligned}
{\hat \Gamma}\Big(\max \mathbb{C}^*\Big) \geq \max \mathbb{C}^*.
\end{aligned}
\end{equation}
As ${\hat \Gamma}$ is increasing as shown in the proof of Theorem~\ref{theorem_equi1}, we have
\begin{equation}
\label{eqn:for16g}
\begin{aligned}
{\hat \Gamma}\bigg({\hat \Gamma}\Big(\max \mathbb{C}^*\Big)\bigg) \geq {\hat \Gamma}\Big(\max \mathbb{C}^*\Big),
\end{aligned}
\end{equation}
and thus ${\hat \Gamma}\Big(\max \mathbb{C}^*\Big) \in \mathbb{C}^*$. Then, we have ${\hat \Gamma}\Big(\max \mathbb{C}^*\Big) \leq \max \mathbb{C}^*_i$ and obtain
\begin{equation}
\label{eqn:for16h}
\begin{aligned}
{\hat \Gamma}\Big(\max \mathbb{C}^*\Big) = \max \mathbb{C}^*.
\end{aligned}
\end{equation}
As a result, $\max \mathbb{C}^*$ is a fixed point of $\Gamma$ which contains $\Big(\boldsymbol{\rho}^*(\phi), \boldsymbol{\xi}^*(\phi)\Big), \forall \phi \in \Phi$. This indicates that the equilibrium exists, i.e., $\Big(\boldsymbol{\hat \rho}(\phi), \boldsymbol{\hat \xi}(\phi)\Big) = \Big(\boldsymbol{\rho}(\phi)^*, \boldsymbol{\xi}^*(\phi)\Big), \forall \phi \in \Phi$.

Since the utilities of all CSs follow the increasing function until reaching a convergence under the $\kappa$, there is no further solution improvement from all the CSs' contracts. Therefore, we conclude that the algorithm must converge to the equilibrium contract solution $\Big(\boldsymbol{\hat \rho}(\phi), \boldsymbol{\hat \xi}(\phi)\Big), \forall \phi \in \Phi$, to ensure that there is no such CS-$i$ can improve its utility, i.e., $U_i({\hat \rho}_i(\phi), {\hat \xi}_i(\phi), \boldsymbol{\hat \rho}_{-i}(\phi), \boldsymbol{\hat \xi}_{-i}(\phi))$, unilaterally.

%
%
%

\end{document}